%% file: index.tex
\RequirePackage{amsmath} 
\documentclass[runningheads,envcountsame]{llncs}

\usepackage[dvipsnames]{xcolor}
\usepackage{comment} 
\usepackage{bussproofs}
\usepackage{mathtools}
\usepackage{amssymb}
\usepackage{xifthen}
\usepackage{tikz}
\usepackage{cancel}
\usepackage{enumerate}
\usepackage{placeins} 
\usepackage{hyperref} 
\usetikzlibrary{decorations.pathmorphing} 
\usetikzlibrary{arrows,automata,calc,positioning,arrows.meta,shapes,matrix}
\usepackage{graphicx} 
\usepackage{diagbox} 
\usepackage[linesnumbered,ruled]{algorithm2e}
\usepackage{multirow}

\newcommand{\todo}[1]{}

\newcommand{\atraces}{\irtraces}
\newcommand{\alts}{\mathcal{AIA}}
\newcommand{\AP}[1]{\langle #1 \rangle} 
\newcommand{\SE}{\operatorname{\mathcal{D}}}
\newcommand{\altsop}{\operatorname{AIA}}
\newcommand{\ltsop}{\operatorname{IA}}

\newcommand{\isdef}[1][0mm]{\mathrel{\hspace{#1}=_\text{}\hspace{#1}}}
\newcommand{\ifdef}[1][0mm]{\mathrel{\hspace{#1}\iff_\text{}\hspace{#1}}}
\newcommand{\ifdefsmall}[1][0mm]{\mathrel{\hspace{#1}\Leftrightarrow_\text{}\hspace{#1}}}
\newcommand{\lts}[1][]{
 \ifthenelse{\equal{#1}{}}
 {\mathcal{IA}}
 {\mathcal{IA}(#1)}
}
\newcommand{\subst}[1]{T_{#1}}
\newcommand{\after}[3][]{#2 \mathrel{\operatorname{after}_{#1}} #3}

\newcommand{\tester}{\operatorname{tester}}
\newcommand{\pass}{\mathbf{pass}}
\newcommand{\fail}{\mathbf{fail}}
\newcommand{\passes}{\mathrel{\mathbf{passes}}}
\newcommand{\fails}{\mathrel{\mathbf{fails}}}

\newcommand{\inp}{\operatorname{in}}
\newcommand{\out}{\operatorname{out}}
\renewcommand{\det}{\operatorname{det}}

\newcommand{\traces}{\operatorname{traces}}
\newcommand{\irtraces}{\operatorname{Ftraces}}

\newcommand{\ttag}[1]{\tag*{\text{[#1]}}}
\newcommand{\qedtag}{\tag*{\qed}}
\newcommand{\ir}{\le_\textit{if}}

\newcommand{\ireq}{\equiv_\textit{if}}
\newcommand{\rcl}{\operatorname{fcl}}
\newcommand{\notrel}[2][3pt]{\mathrel{\hspace{#1}\cancel{\hspace{-#1}#2\hspace{-#1}}\hspace{#1}}}

\newcommand{\iftracesdom}{\mathcal{FT}}
\newcommand{\te}{\operatorname{\rceil|}}

\makeatletter%
\@ifclassloaded{llncs}%
  {}%
  {%
    \newtheorem{definition}{Definition}[section]%
    \newtheorem{lemma}{Lemma}[section]%
    \newtheorem{theorem}{Theorem}[section]%
    \newtheorem{proposition}{Proposition}[section]%
    \newtheorem{corollary}{Corollary}[section]%
    \newtheorem{example}{Example}[section]%
    \newenvironment{remark}{\paragraph{Remark}}{$\hfill\triangle$}%
    \newenvironment{proof}{\paragraph{Proof}}{\hfill$\square$\\}%
  }%
\makeatother%

\newcommand{\Outc}[1][]{
  \ifthenelse{\equal{#1}{}}
    {\mathop{\operatorname{outc}}}
    {\mathop{\operatorname{outc}}(\stratenv^{#1},\stratsys^{#1},\stratdet^{#1},\stratrc^{#1})}
}
\newcommand{\Outcati}[1][]{
  \ifthenelse{\equal{#1}{}}
    {\mathop{\operatorname{outc}}}
    {\mathop{\operatorname{outc}}(\stratenvati^{#1},\stratsys^{#1},\stratdet^{#1},\stratrc^{#1})}
}

\tikzstyle{every state}=[initial text=,minimum size=5mm,inner sep=0pt,shape=ellipse,draw=black]
\tikzstyle{top}=[initial text=,minimum size=0.15cm,inner sep=0pt,shape=ellipse,execute at begin node=$\top$]
\tikzset{trans/.style={>={Latex[length=2mm]},-{Latex[length=2mm]}}}
\tikzset{every initial by arrow/.append style={trans}}
\tikzstyle{emptystate}=[state]
\tikzstyle{verdictstate}=[state,minimum height=5mm]

\newcommand{\tikzlts}{
 \tikzset{every loop/.append style={trans}}
 \tikzset{every edge/.append style={trans}}
}

\tikzstyle{recstate}=[state,shape=rectangle,rounded corners=6pt,draw=black,inner sep=2pt]

\newif\ifhideproofs
\newif\ifshowshortversion

\newenvironment{longversion}{\ifvmode\else\unskip\fi}{\ignorespacesafterend}
\newenvironment{shortversion}{\ifvmode\else\unskip\fi}{\ignorespacesafterend}
\ifshowshortversion
  \excludecomment{longversion}
\else
  \excludecomment{shortversion}
\fi

\ifhideproofs
  \usepackage{environ}
  \NewEnviron{hidep}{}

\fi

\title{Combining Partial Specifications using\\Alternating Interface Automata\thanks{Funded by the Netherlands Organisation of Scientific Research (NWO-TTW), project 13859: SUMBAT - SUpersizing Model-BAsed Testing}}%
\author{Ramon Janssen}%
\institute{Radboud University, Nijmegen\\\email{ramonjanssen@cs.ru.nl}}%
\authorrunning{R. Janssen}%

\begin{document}

\maketitle

\input{abstract}

\input{introduction}

\input{preliminaries}
\input{input-refusals}
\input{alts}

\input{alts-ir}
\input{alts-determinization}
\input{alts-lts}

\input{test-cases}
\input{conclusion}

\subsection*{\ackname}
We thank Jan Tretmans and Frits Vaandrager for their valuable feedback.

\FloatBarrier

\bibliographystyle{plain}
\bibliography{bib}{}

\end{document}

%% file: abstract.tex
\begin{abstract}
  To model real-world software systems, modelling paradigms should support a form of compositionality.
  In interface theory and model-based testing with inputs and outputs, \emph{conjunctive} operators have been introduced: the behaviour allowed by composed specification $s_1 \wedge s_2$ is the behaviour allowed by both partial models $s_1$ and $s_2$.
  The models at hand are non-deterministic \emph{interface automata}, but the interaction between non-determinism and conjunction is not yet well understood.
  On the other hand, in the theory of \emph{alternating automata}, conjunction and non-determinism are core aspects.
  Alternating automata have not been considered in the context of inputs and outputs, making them less suitable for modelling software interfaces.
  In this paper, we combine the two modelling paradigms to define \emph{alternating interface automata} (AIA).
  We equip these automata with an observational, trace-based semantics, and define testers, to establish correctness of black-box interfaces with respect to an AIA specification.
\end{abstract}


%% file: introduction.tex
\section{Introduction}
\label{sec:introduction}


The challenge of software verification is to ensure that software systems are correct, using techniques such as model checking and model-based testing.
To use these techniques, we assume that we have an abstract specification of a system, which serves as a description of what the system should do.
A popular approach is to model a specification as an automaton.
However, the huge number of states in typical real-world software systems quickly makes modelling with explicit automata infeasible.
A form of compositionality is therefore usually required for scalability, so that a specification can be decomposed into smaller and understandable parts.
Parallel composition is based on a structural decomposition of the modelled system into components, and it thus relies on the assumption that components themselves are small and simple enough to be modelled.
This assumption is not required for logical composition, in which partial specification models of the same component or system are combined in the manner of logical conjunction.
Formally, for a composition to be conjunctive, the behaviour allowed by $s_1 \wedge s_2$ is the behaviour allowed by both partial specifications $s_1$ and $s_2$.
Such a composition is important for scalability of modelling, as it allows writing independent partial specifications, sometimes called view modelling~\cite{Benetal15}.
On a fundamental level, specifications can be seen as logical statements about software, and the existence of conjunction on such statements is only natural.
Conjunctive operators have been defined in many language-theoretic modelling frameworks, such as for regular expressions~\cite{NaYa60} and process algebras~\cite{Bri90}.

\subsection{Conjunction for Inputs and Outputs}

A conjunctive operator $\wedge$ has also been introduced in many automata frameworks for formal verification and testing, such as interface theory~\cite{Doetal08}, ioco theory~\cite{Benetal15} and the theory of substitutivity refinement~\cite{Chietal14}.
Within these theories, systems are modelled as \emph{labelled transition systems}~\cite{Tre08} or \emph{interface automata}~\cite{AlHe01} (IA), and actions are divided into inputs and outputs.

An informal example of some (partial) specification models, as could be expressed in these theories, is shown by the automata in Figure~\ref{fig:intro-coffee}, in which inputs are labelled with question marks, and outputs with exclamation marks.
The specifications represent a vending machine with two input buttons (?a and ?b), which provides coffee (!c) and tea (!t) as outputs, optionally with milk (!c+m and !t+m).
The first model, $p$, specifies that after pressing button ?a, the machine dispenses coffee.
The second model, $q$, specifies that after pressing button ?b, the machine has a choice between dispensing tea, or tea with milk.
The third model, $r$, is similar, but uses non-determinism to specify that button ?b results in coffee with milk or tea with milk.

The fourth model, $p \wedge q \wedge r$, states that all former three partial models should hold.
Here, we use the definition of $\wedge$ from~\cite{Benetal15}, but the definition from~\cite{Chietal14} is similar.
An input is specified in the combined model if it is specified in any partial model, making both buttons ?a and ?b specified.
Additionally, an output is allowed in the combined model if it is allowed by all partial models, meaning that after button ?b, only tea with milk is allowed.

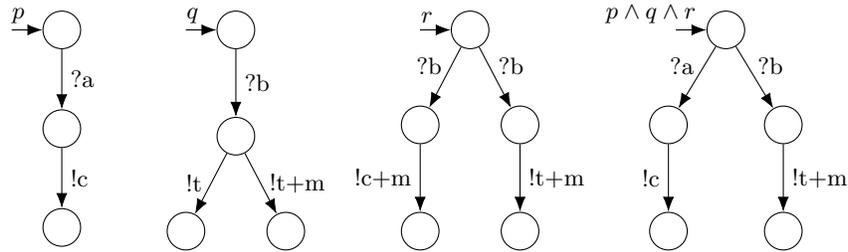
\begin{figure}
  \begin{center}
    \begin{tikzpicture}[->,node distance=11mm]
      \tikzlts
      
      \node[state,initial=left] (A1) {};
      \node (init-text) [above left=-2mm and 2mm of A1] {$p$};
      \node[state] (A2) [below=8mm of A1] {};
      \node[state] (A3) [below=8mm of A2] {};
      \path
      (A1) edge node [right] {?a} (A2)
      (A2) edge node [right] {!c} (A3)
      ;
      \node[state,initial=left] (C1) [right=18mm of A1] {};
      \node (init-text) [above left=-2mm and 2mm of C1] {$q$};
      \node[state] (C2) [below=9mm of C1] {};
      \node[state] (C3) [below left=9mm and 3mm of C2] {};
      \node[state] (C4) [below right=9mm and 3mm of C2] {};
      \path
      (C1) edge node [right] {?b} (C2)
      (C2) edge node [left] {!t} (C3)
      (C2) edge node [right] {!t+m} (C4)
      ;
      \node[state,initial=left] (B1) [right=26mm of C1] {};
      \node (init-text) [above left=-2mm and 2mm of B1] {$r$};
      \node[state] (B3) [below left=9mm and 3mm of B1] {};
      \node[state] (B4) [below right=9mm and 3mm of B1] {};
      \node[state] (B5) [below=9mm of B3] {};
      \node[state] (B6) [below=9mm of B4] {};
      \path
      (B1) edge node [above left=-1mm] {?b} (B3)
      (B1) edge node [above right=-1mm] {?b} (B4)
      (B3) edge node [left] {!c+m} (B5)
      (B4) edge node [right] {!t+m} (B6)
      ;
      \node[state,initial=left] (A1B1C1) [right=60mm of C1] {};
      \node (init-text) [above left=-2mm and 1mm of A1B1C1] {$p \wedge q \wedge r$};
      \node[state] (A2) [below left=9mm and 4mm of A1B1C1] {};
      \node[state] (A3) [below=9mm of A2] {};
      \node[state] (B2C2) [below right=9mm and 4mm of A1B1C1] {};
      \node[state] (B3C3) [below=9mm of B2C2] {};
      \path
      (A1B1C1) edge node [above left=-1mm] {?a} (A2)
      (A2) edge node [left] {!c} (A3)
      (A1B1C1) edge node [above right=-1mm] {?b} (B2C2)
      (B2C2) edge node [right] {!t+m} (B3C3)
      ;
    \end{tikzpicture}
    \caption{Three independent specifications for a vending machine, and their conjunction.}
    \label{fig:intro-coffee}
  \end{center}
\end{figure}
\vspace{-10mm}

\subsection{Conjunctions of states}

This form of conjunctive composition acts as an operator on entire models.
However, a partial specification could also describe the expected behaviour of a particular state of the system, other than the initial state.
For example, suppose that the input ?on turns the vending machine on, after which the machine should behave as specified by $p$, $q$ and $r$ from Figure~\ref{fig:intro-coffee}.
This, by itself, is also a specification, illustrated by $s$ in Figure~\ref{fig:intro-coffee-not-initial}.
However, the formal meaning of this model is unclear: transitions connect states, whereas $p \wedge q \wedge r$ is not a state but an entire automaton.
A less trivial case is partial specification $t$, also in Figure~\ref{fig:intro-coffee-not-initial}: after obtaining any drink by input ?take, we should move to a state where we can obtain a drink as described by specifications $p$, $q$, $r$ and $t$.
Thus, we combine conjunctions with a form of recursion.
This cannot easily be formalized using $\wedge$ as an operator on automata, like in~\cite{Benetal15,Chietal14,Doetal08}.
Defining conjunction as a composition on individual states would provide a formal basis for these informal examples.

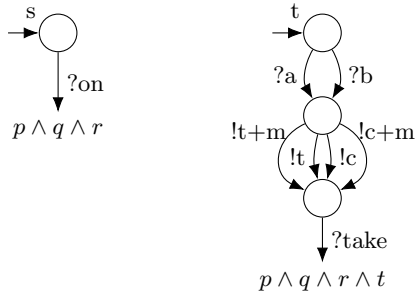
\begin{figure}
  \begin{center}
    \begin{tikzpicture}[->,node distance=11mm]
      \tikzlts
      
      \node[emptystate,initial=left] (A1) {};
      \node (init-text) [above left=-1mm and 0 of A1] {s};
      \node[] (A2) [below=8mm of A1] {$p \wedge q \wedge r$};
      \path
      (A1) edge node [right] {?on} (A2)
      ;
      \node[initial=left,emptystate] (D1) [right=30mm of A1] {};
      \node (init-text) [above left=-1mm and 0 of D1] {t};
      \node[emptystate] (D2) [below of=D1] {};
      \node[emptystate] (D3) [below of=D2] {};
      \node[] (D4) [below of=D3] {$p \wedge q \wedge r \wedge t$};
      \path
      (D1) edge [bend left] node [right] {?b} (D2)
      (D1) edge [bend right] node [left] {?a} (D2)
      (D2) edge [bend left=14] node [right] {!c} (D3)
      (D2) edge [bend left=65,looseness=1.4] node [above right=1mm and -2mm] {!c+m} (D3)
      (D2) edge [bend right=14] node [left] {!t} (D3)
      (D2) edge [bend right=65,looseness=1.4] node [above left=1mm and -2mm] {!t+m} (D3)
      (D3) edge node [right] {?take} (D4)
      ;
    \end{tikzpicture}
    \caption{Two specifications with transitions to a conjunction.}
    \label{fig:intro-coffee-not-initial}
  \end{center}
\end{figure}


Conjunctions of states are a main ingredient of \emph{alternating automata}~\cite{ChKoSt81}, in which conjunctions and non-determinism alternate.
Here, non-determism acts as logical disjunction, dually to conjunction.
Because of this duality, both conjunction and disjunction are treated analogously: both are encoded in the transition relation of the automaton.
This contrasts the approach of defining conjunction directly on IAs, where non-determinism is encoded in the transition relation of the IA, whereas conjunction is added as an operator on IAs, leaving the duality between the two unexploited.
In fact, the conjunction-operator in~\cite{Benetal15} even requires that any non-determinism in its operands is removed first, by performing an exponential determinization step.
For example, model $r$ in Figure~\ref{fig:intro-coffee} is non-deterministic, and must be determinized to the form of model $q$ before $p \wedge q \wedge r$ is computed.
This indicates that it is hard to combine conjunction and non-determinism in an elegant way, without understanding their interaction.

Despite their inherent support for conjunction, alternating automata are not entirely suitable for modeling the behaviour of software systems,
since they lack the distinction between inputs and outputs.
In this respect, alternating automata are similar to deterministic finite automata (DFAs).
Distinguishing inputs and outputs in an IA allows modelling of software systems in a less abstract way than with the homogeneous alphabet of actions of DFAs and alternating automata.

\subsection{Contributions}

We combine concepts from the worlds of interface theory and alternating automata, leading to Alternating Interface Automata (AIAs), and show how these can be used in the setting of a trace semantics for observable inputs and outputs.
We provide a solid formal basis of AIAs, by
\begin{itemize}
  \item
  combining alternation with inputs and outputs (Section~\ref{sec:aia}),
  \item
  defining a trace semantics for AIAs (Section~\ref{sec:alts-ir}), by lifting the \emph{input-failure refinement} semantics for non-deterministic interface automata~\cite{JaVaTr19} to AIAs,
  \item
  providing insight into the semantics of an AIA, by defining a determinization operator (Section~\ref{sec:determinization}) and a transformation between IAs and AIAs (Section~\ref{sec:alts-lts}), and
  \item
  defining testers (Section~\ref{sec:testers}), which represent practical testing scenarios for establishing input-failure refinement between a black-box implementation IA and a specification AIA, analogously to ioco test case generation~\cite{Tre08}.
\end{itemize}

The definition of input-failure refinement~\cite{JaVaTr19} is based upon the observation that, for a non-deterministically reached set of states $Q$, the observable outputs of that set are the union of the outputs of the individual states in $Q$, whereas the specified inputs for $Q$ are the intersection of the inputs specified in individual states in $Q$.
For conjunction, we invert this: outputs allowed by a conjunction of states are captured by the intersection, whereas specified inputs are captured by the union.
In this way, our AIAs seamlessly combine the duality between conjunction and non-determinism with the duality between inputs and outputs.

\begin{shortversion}
  Proofs can be found in the technical report~\cite{??}.
\end{shortversion}

%% file: preliminaries.tex
\section{Preliminaries}
\label{sec:preliminaries}

We first recall the definition of interface automata~\cite{AlHe01} and input-failure refinement~\cite{JaVaTr19}.
The original definition of IAs~\cite{AlHe01} allows at most one initial state, but we generalize this to sets of states.
Moreover, \cite{AlHe01} supports
internal actions, which we do not need.
Transitions are commonly encoded by a relation, whereas we use a function.

\begin{definition}
  \label{def:ia}
  An Interface Automaton (IA) is a 5-tuple $(Q,I,O,T,Q^0)$, where
  \begin{itemize}
    \item $Q$ is a set of states,
    \item $I$ and $O$ are disjoint sets of input and output actions, respectively,
    \item 
    $T : Q \times (I \cup O\}) \rightarrow \mathcal{P}(Q)$ is an image-finite transition function (meaning that $T(q,\ell)$ is finite for all $q$ and $\ell$), and
    \item $Q^0 \subseteq Q$ is a finite set of initial states.
  \end{itemize}
  
  The domain of IAs is denoted $\lts$.
  For $s \in \lts$, we refer to its respective elements by $Q_s$, $I_s$, $O_s$, $T_s$, $Q_s^0$. 
  For $s_1, s_2,\ldots, s_A, s_B, \dots$ a family of IAs, we write $Q_j$, $I_j$, $O_j$, $T_j$ and $Q_j^0$ to refer to the respective elements, for $j = 1, 2,\ldots, A, B, \dots$.
  
\end{definition}

In examples, we represent IAs graphically as in Figure~\ref{fig:intro-coffee}.
For the remainder of this paper, we assume fixed input and output alphabets $I$ and $O$ for IAs, with $L = I \cup O$.
For (sets of) sequences of actions, $*$ denotes the Kleene star, and $\epsilon$ denotes the empty sequence.
We define auxiliary notation in the style of~\cite{Tre08}.

\renewcommand{\xRightarrow}{\xrightarrow}
\begin{definition}
  \label{def:lts-definitions}
  Let $s \in \lts$, $Q \subseteq Q_s$, $q,q' \in Q_s$, 
  $\ell \in L$ and $\sigma \in L^*$.
  We define
  \begin{align*}
    q \xrightarrow{\epsilon}_s q'\:\: &\ifdefsmall q = q' &
    \hspace{-4mm}q \xrightarrow{\sigma\ell}_s q'\:\: &\ifdefsmall \exists r \in Q_s: q \xrightarrow{\sigma}_s r \wedge q' \in T_s(r,\ell)\\
    q \xRightarrow{\sigma}_s \:\: &\ifdefsmall \exists r \in Q_s: q \xRightarrow{\sigma}_s r &
    q \not\xRightarrow{\sigma}_s\:\: &\ifdefsmall\:\: \neg (q \xrightarrow{\sigma}_s)\\
    \traces_s(q) &\isdef \{\sigma \in L^* \mid q \xRightarrow{\sigma}_s\} &
    \after[s]{Q}{\sigma} &\isdef \{r \in Q_s \mid \exists r' \in Q: r' \xRightarrow{\sigma}_s r\} \\
    \traces(s) &\isdef \bigcup_{q \in Q_s^0} \traces_s(q) &
    \after{s}{\sigma} &\isdef \after[s]{Q_s^0}{\sigma}\\
    \out_s(Q) &\isdef \{x \in O \mid \exists q \in Q: q \xRightarrow{x}_s \} &
    \inp_s(Q) &\isdef \{a \in I \mid \forall q \in Q: q \xRightarrow{a}_s\}
  \end{align*}
  \vspace{-8mm}
  \begin{align*}
    \text{$q$ is a \emph{sink-state} of $s$} &\ifdef \forall \ell \in L: T_s(q,\ell) \subseteq \{q\}\\
    \text{$s$ is \emph{input-enabled}} &\ifdef \forall q \in Q_s: \inp_{s}(q) = I\\
    \text{$s$ is \emph{deterministic}} &\ifdef \forall \sigma \in L^*, |\after{s}{\sigma}| \le 1
  \end{align*}
\end{definition}

We omit the subscript for interface automaton $s$ when clear from the context.

%% file: input-refusals.tex

We use IAs to represent black-box systems, which can produce outputs, and consume or refuse inputs from the environment.
This entails a notion of observable behaviour, which we define in terms of \emph{input-failure traces}~\cite{JaVaTr19}.

\begin{definition}
  \label{def:refusal}
  For any input action $a$, we denote the \emph{input-failure of $a$} as $\overline a$.
  Likewise, for any set of inputs $A$, we define $\overline A \isdef \{\overline a \mid a \in A\}$.
  The domain of \emph{input-failure traces} is defined as 
  $\iftracesdom_{I,O} \isdef[0mm] L^* \cup L^* \cdot \overline I$. 
  For $s \in \lts$, we define
  \[
    \irtraces(s) \isdef \traces(s) \cup \{\sigma \overline a \mid \sigma \in L^*, 
    a \in I, a \not \in \inp(\after{s}{\sigma})
    \}
  \]
\end{definition}

Thus, a trace $\sigma\overline{a}$ indicates that $\sigma$ leads  to a state where $a$ is not accepted, e.g. a greyed-out button which cannot be clicked.

Any such set of input-failure traces is prefix-closed.
Input-failure traces are the basis of \emph{input-failure refinement}, which we will now explain briefly.
This refinement relation was introduced in~\cite{JaVaTr19} to bridge the gap between alternating refinements~\cite{AlHe01,Aluetal98} and ioco theory~\cite{Tre08}.
Similarly to normal trace inclusion, the idea is that an implementation may only show a trace if a specification also shows this trace.
Moreover, the most permissive treatment of an input is to fail it, so if a specification allows an input failure, then it also must allow acceptance of that input, as expressed by the \emph{input-failure closure}.

\begin{definition}
  \label{def:ir-inclusion}
  Set $S \subseteq \iftracesdom_{I,O}$ of input-failure traces is \emph{input-failure closed} if, for all $\sigma \in L^*$, $a \in I$ and $\rho \in \iftracesdom_{I,O}$,
  $ \sigma\overline a \in S \implies \sigma a \rho \in S$.
  The \emph{input-failure closure} of $S$ is the smallest input-failure closed superset of $S$, that is, $\rcl(S) \isdef S \cup \{\sigma a \rho \mid \sigma\overline a \in S, \rho \in \iftracesdom_{I,O}\}$.

  \emph{Input-failure refinement} and \emph{input-failure equivalence} on IAs are respectively defined as
  \begin{align*}
    s_1 \ir s_2 &\ifdef \irtraces(s_1) \subseteq \rcl(\irtraces(s_2))\text{, and}\\
    s_1 \ireq s_2 &\ifdef s_1 \ir s_2 \wedge s_2 \ir s_1.
  \end{align*}
\end{definition}

The input-failure closure of the $\irtraces$ serves as a canonical representation of the behaviour of an IA.
That is, two models are input-failure equivalent if and only if the closure of their input-failure traces is the same, as stated in Proposition~\ref{pro:rcl-irtraces-canonical}.

\begin{proposition}
  \cite{JaVaTr19}
  \label{pro:rcl-irtraces-canonical}
  Let $s_1, s_2 \in \lts$.
  Then 
  \begin{align*}
    s_1 \ir s_2 \iff& \rcl(\irtraces(s_1)) \subseteq \rcl(\irtraces(s_2))\\
    s_1 \ireq s_2 \iff& \rcl(\irtraces(s_1)) = \rcl(\irtraces(s_2))
  \end{align*}
\end{proposition}

Proposition~\ref{pro:rcl-irtraces-canonical} implies that relation $\ir$ is reflexive ($s \ir s$) and transitive ($s_1 \ir s_2 \wedge s_2 \ir s_3 \implies s_1 \ir s_3$).
Formally, it is thus a preorder, making it suitable for stepwise refinement.

%% file: alts.tex
\section{Alternating Interface Automata}

Real software systems are always in a single state, but the precise state of a system cannot always be derived from an observed trace.
Due to non-determinism, a trace may lead to multiple states.
In IAs, this is modelled as a set of states, such as the set of initial states, the set $T(q,\ell)$ for state $q$ and action $\ell$, and the set $\after{s}{\sigma}$ for IA $s$ and trace $\sigma$.
The domain of such non-deterministic views on an IA with states $Q$ is thus the powerset of states, $\mathcal{P}(Q)$.
In set of states $Q$, traces from any individual state in $Q$ may be observed.


\subsection{Alternation}
Alternation generalizes this view on automata: a system may not only be non-deterministically in multiple states, but also conjunctively.
When conjunctively in multiple states, only traces which are in \emph{all} these states may be observed.
Alternation is formalized by exchanging the domain $\mathcal{P}(Q)$ for the domain $\SE(Q)$.
Formally, $\mathcal{D}(Q)$ is the \emph{free distributive lattice}, which exist for any set $Q$~\cite{Pri90}.

\newcommand{\qmid}{\;\mid\;}
\begin{definition}
  \label{def:free-distributive-lattice}
  For any set $Q$, $\SE(Q)$ denotes the \emph{free distributive lattice generated by $Q$}.
  That is, $\SE(Q)$ is the domain of equivalence classes of terms, inductively defined by the the grammar
  \[e \quad=\quad \top \qmid \bot \qmid \AP{q} \qmid e_1 \vee e_2 \qmid e_1 \wedge e_2 \hspace{12mm} \text{with $q \in Q$,}\]
  where equivalence of terms is completely defined by the following axioms:

  \def\arraystretch{1.3}
  \noindent
  \begin{tabular}{ccr}
    $e_1 \vee e_2 = e_2 \vee e_1$
      & $e_1 \wedge e_2 = e_2 \wedge e_1$
      & \hspace{-3mm}[Commutativity]\\
    $e_1 \vee (e_2 \vee e_3) = (e_1 \vee e_2) \vee e_3$
      & \hspace{4mm} $e_1 \wedge (e_2 \wedge e_3) = (e_1 \wedge e_2) \wedge e_3$
      & [Associativity]\\
    $e_1 \vee (e_1 \wedge e_2) = e_1$
      & $e_1 \wedge (e_1 \vee e_2) = e_1$
      & [Absorption]\\
    $e \vee e = e$
      & $e \wedge e = e$
      & [Idempotence]\\
    \multicolumn{3}{l}{$e_1 \vee (e_2 \wedge e_3) = (e_1 \vee e_2) \wedge (e_1 \vee e_3)$ \qquad $e_1 \wedge (e_2 \vee e_3) = (e_1 \wedge e_2) \vee (e_1 \wedge e_3)$} \\
    &&[Distributivity]\\
    $e \vee \top = \top$
      & $e \wedge \bot = \bot$
      & [Identity]
  \end{tabular}
  In short, ($\SE(Q)$, $\vee$, $\wedge$, $\bot$, $\top$) forms a distributive lattice.
  Expression $\AP{q}$ is named the embedding of $q$ in $\SE(Q)$, and operators $\vee$ and $\wedge$ are named disjunction and conjunction, respectively.
  For the remainder of this paper, we make no distinction between expressions and their equivalence classes.
  
  For finite $n$, we introduce the shorthand $n$-ary operators $\bigvee$ and $\bigwedge$, as follows:
  \begin{align*}
    \bigvee \{e_1, e_2, \dots e_n\} &\isdef e_1 \vee e_2 \vee \dots e_n
      & \bigvee \emptyset &\isdef \bot\\
    \bigwedge \{e_1, e_2, \dots e_n\} &\isdef e_1 \wedge e_2 \wedge \dots e_n
      & \bigwedge \emptyset &\isdef \top
  \end{align*}
\end{definition}

\begin{longversion}
  \begin{remark}
    Identifying expressions and their equivalence classes requires that the equivalence relation on expressions is a congruence for all functions that we define on $\SE(Q)$.
    We do not write out explicit checks for congruence.
    \todo{does this need more explanation?}
  \end{remark}
\end{longversion}

We distinguish the embedding $\AP{q} \in \SE(Q)$ from $q$ itself.
We require this distinction only in Definition~\ref{def:alts-determinization}, where we will point this out.
Otherwise, we do not need this distinction, so we write $q$ instead of $\AP{q}$.

Intuitively, disjunction $q_1 \vee q_2$ replaces the non-deterministic set $\{q_1,q_2\}$.
This is formalized by extending IAs with alternation.
\label{sec:aia}


\begin{definition}
 \label{def:alts}
 An \emph{alternating interface automaton} (AIA) is defined as a 5-tuple $(Q,I,O,T,e^0)$ where
 \begin{itemize}
  \item $Q$ is a set of states, and elements of $\SE(Q)$ are referred to as \emph{configurations},
  \item $I$ and $O$ are disjoint sets of input and output actions, respectively, 
  \item $T : Q \times (I \cup O) \rightarrow \SE(Q)$ is a transition function, with $T(q,a) \neq \bot$ for all $a \in I$, and
  \item $e^0 \in \SE(Q)$ is the initial configuration.
 \end{itemize}
 
 The domain of AIAs is denoted by $\alts$.
 Notations for IAs are reused for AIAs, if this causes no ambiguity.
 For $\ell \in L$, we define $\subst{\ell} : Q \rightarrow \SE(Q)$ by $\subst{\ell}(q) = T(q,\ell)$.
\end{definition}


Configurations $\top$ and $\bot$ are analogous to the empty set of states in an IA $s$: if $T_s(q,\ell) = \emptyset$, this means that state $q$ does not have a transition for $\ell$.
In terms of input-failure refinement, not having a transition for an input means that the input is underspecified, whereas not having a transition for an output means that the output is forbidden.
This distinction is made explicit in AIA by using $\top$ to represent underspecification and $\bot$ to represent forbidden behaviour.
We will formalize this in Section~\ref{sec:alts-ir}.
Definition~\ref{def:alts} also allows output transitions to $\top$, meaning that  the behaviour is unspecified after that output.
Automata models which do not allow distinct configurations $\top$ and $\bot$ commonly represent such underspecified behaviour with an explicit chaotic state~\cite{Benetal15,BiReTr04} instead.

We graphically represent AIAs in a similar way as IAs, with some additional rules.
A transition $T(q^0,\ell) = \AP{q^1}$ is represented by a single arrow from $q^0$ to $q^1$.
We represent $T(q^0,\ell) = q^1 \vee q^2$ by two arrows $q^0 \xrightarrow{\ell} q^1$ and $q^0 \xrightarrow{\ell} q^2$, analogous to non-determinism in IAs.
Conjunction $T(q^0,\ell) = q^1 \wedge q^2$ is shown by adding an arc between the arrows.
Nested expressions are represented by successive splits, as shown in Example~\ref{exa:aia}.
A state $q$ without outgoing arrow for an output $\ell \in O$ represents $T(q,\ell) = \bot$, and a state without input transitions for input $\ell$ indicates $T(q,\ell) = \top$.
For $\ell \in O$, a transitions $T(q,\ell) = \top$ is shown with an arrow to $\top$, denoting underspecification, but note that $\top$ is a configuration, not a state.

\begin{example}
  \label{exa:aia}
  Figure~\ref{fig:alts-example} shows AIA $s_A$, with $Q_A = \{q_A^0, q_A^1, q_A^2\}$, $I = \{\text{?a},\text{?b}\}$, $O = \{\text{!x},\text{!y}\}$, $e_A^0 = q_A^0$ and $T$ given by the following table:
  
  \setlength{\tabcolsep}{3mm}
    \begin{tabular}{|r | c c c c |}\hline
    \diagbox[width=6em]{state}{action} & ?a & ?b & !x & !y \\ \hline
    $q_A^0$ & $q_A^0 \wedge (q_A^1 \vee q_A^2)$ & $\top$ & $q_A^0$ & $q_A^0$ \\
    $q_A^1$ & $\top$ & $\top$ & $\top$ & $\bot$ \\
    $q_A^2$ & $\top$ & $q_A^0$ & $\bot$ & $q_A^2$ \\\hline
  \end{tabular}
  
  Moreover, AIA $s_B$ combines the partial specifications from Section~\ref{sec:introduction}.
  
  \begin{figure}
    \strut
    \hspace{-20mm}
    \tikzlts
    \begin{tikzpicture}[node distance=16mm]
      \node [state] [initial,initial text=] (0) {$q_A^0$};
      \node (init-text) [above left=-2mm and 0.5mm of 0] {$s_A$};
      \node (0b) [right of=0] {};
      \node [state] (1) [above right of=0b] {$q_A^1$};
      \node [state] (2) [below right of=0b] {$q_A^2$};
      \node [top] (3) [right of=1] {};
      \node (dummy) [below right=7mm and 10.5mm of 0] {};
      \draw[-] ($(0)!0.45!(0b)$) to[bend left=30] ($(0)!0.35!(dummy)$);
      \path
      (0) edge [-] node [below] {?a} (0b.center)
      (0) edge [loop above,in=60,out=85,looseness=10] node {\vphantom{!y}!x} (0)
      (0) edge [loop above,in=120,out=95,looseness=10] node {!y} (0)
      (0.east) edge [in=-70,out=-45,looseness=5] node [below] {} (0)
      (0b.center) edge node [above left] {} (1)
      (0b.center) edge node [below left=0mm and -2mm] {} (2)
      (1) edge [below] node {!x} (3)
      (2) edge [loop right] node {!y} (2)
      (2) edge [out=180,in=-110] node [below] {?b} (0)
      ;

      \tikzset{node distance=11mm}
      \node[initial left,state,inner sep=1pt, below left=15mm and 3mm of 0] (init) {$q_B^0$};
      \node (init-text) [above left=-2mm and 0.5mm of init] {$s_B$};
      \node[state] (A1) [below left=11mm and 18mm of init] {$q_B^1$};
      \node[state] (A2) [below of=A1] {$q_B^2$};
      \node[top] (A3) [below of=A2] {};
      \path
      (A1) edge node [right] {?a} (A2)
      (A2) edge node [right] {!c} (A3)
      ;
      \node[state] (C1) [below=9mm of init] {$q_B^3$};
      \node[state] (C2) [below=7mm of C1] {$q_B^4$};
      \node[top] (C3) [below left=7mm and 3mm of C2] {};
      \node[top] (C3b) [below right=7mm and 3mm of C2] {};
      \path
      (C1) edge node [left] {?b} (C2)
      (C2) edge node [left=0.5mm] {!t} (C3)
      (C2) edge node [right=0.5mm] {!t+m} (C3b)
      ;
      \node[state] (B1) [below right=11mm and 30mm of init] {$q_B^5$};
      \node[state] (B3) [below left=7mm and 3mm of B1] {$q_B^6$};
      \node[state] (B4) [below right=7mm and 3mm of B1] {$q_B^7$};
      \node[top] (B5) [below=7mm of B3] {};
      \node[top] (B6) [below=7mm of B4] {};
      \path
      (B1) edge node [left] {?b} (B3)
      (B1) edge node [right] {?b} (B4)
      (B3) edge node [left] {!c+m} (B5)
      (B4) edge node [right] {!t+m} (B6)
      ;
      \node[state] (D1) [below right=11mm and 62mm of init] {$q_B^8$};
      \node[state] (D2) [below of=D1] {$q_B^9$};
      \node[state] (D3) [below of=D2] {$q_B^{10}$};
      \path
      (D1) edge [bend left] node [right] {?b} (D2)
      (D1) edge [bend right] node [left] {?a} (D2)
      (D2) edge [bend left=14] node [right] {!c} (D3)
      (D2) edge [bend left=70,looseness=2] node [above right=1mm and -2mm] {!c+m} (D3)
      (D2) edge [bend right=14] node [left] {!t} (D3)
      (D2) edge [bend right=70,looseness=2] node [above left=1mm and -2mm] {!t+m} (D3)
      ;
      \path
      (init) edge node [above left=-1.4mm] {} (A1)
      (init) edge node [above right=-1mm and -0.6mm] {} (B1)
      (init) edge node [above left=-1mm and 0mm] {?on} (C1)
      (init) edge node [above right=-1mm] {} (D1)
      ;
      \draw[-] ($(init)!0.2!(A1)$) to[bend right=45] ($(init)!0.1!(D1)$);
      \path
      (D3) edge [out=-155,in=-130,looseness=1.2] (A1)
      (D3) edge [out=-180,in=-15,looseness=1.6] (B1)
      (D3) edge [out=-170,in=-50,looseness=1.6] (C1)
      (D3) edge [out=-20,in=0,looseness=2.3] node [below left=11mm and 13mm] {?take} (D1)
      ;
      \node (aux1) [above left=-2mm and 20mm of D3] {};
      \node (aux2) [below right=5mm and 20mm of D3] {};
      \draw[-] ($(D3)!0.3!(aux1)$) to[bend right=45] ($(D3)!0.25!(aux2)$);
    \end{tikzpicture}
    \vspace{-15mm}
    \caption{Example AIAs $s_A$ and $s_B$.}
    \label{fig:alts-example}
    \label{fig:ALTS-coffee}
  \end{figure}
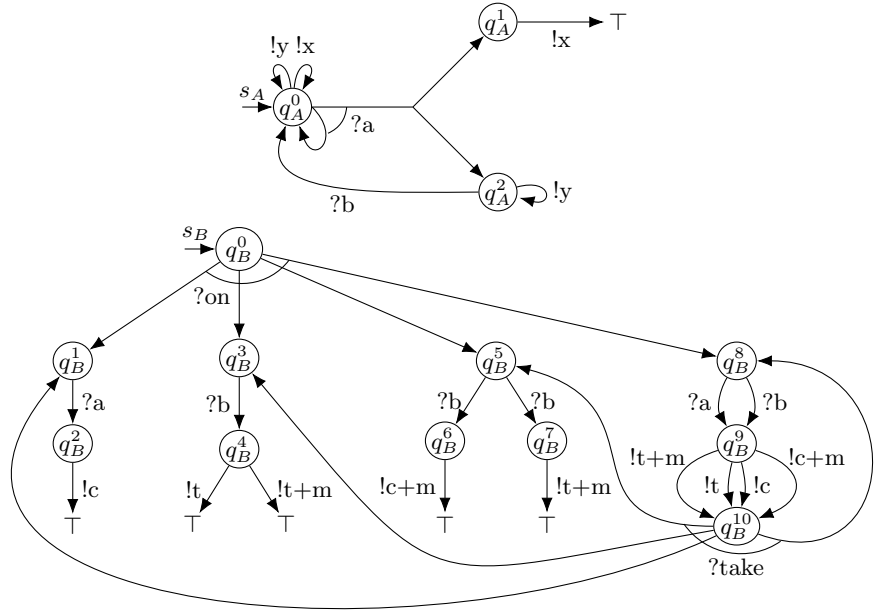
\end{example}

Before defining trace semantics for AIAs, we extend the transition function from single actions to sequences of actions, by defining an $\after{}{}$-function on AIAs.
This function transforms configurations by substituting every state according to the transition function, similarly to the approach for alternating automata in~\cite{ChKoSt81}.

\begin{definition}
  \label{def:substitution}
  Let $f : Q \rightarrow \SE(Q)$ and $e \in \SE(Q)$.
  Then \emph{substitution} $e[f]$ is equal to $e$ with all atomic propositions replaced by $f(e)$.
  Formally, $[f] : \SE(Q) \rightarrow \SE(Q)$ is a postfix operator defined by\\
  \begin{tabular}{cc}
    $(e_1 \vee  e_2) [f] = e_1[f] \vee e_2[f]$ &\quad $(e_1 \wedge  e_2) [f] = e_1[f] \wedge e_2[f]$\\
    \multicolumn{2}{c}{\begin{tabular}{ccc}$\top[f] = \top$ &\quad $\bot[f] = \bot$ &\quad $\AP{q}[f] = f(q)$\end{tabular}}
  \end{tabular}
\end{definition}


\begin{definition}
 \label{def:sla-after}
 Let $s \in \alts$.
 We define $\after{}{} : \SE(Q_s) \times L^* \rightarrow \SE(Q_s)$ as
 \[
  \after[s]{e}{\epsilon} = e \hspace{1cm}
  \after[s]{e}{(\ell \cdot \sigma)} = \after[s]{e[\subst{\ell}]\,}{\sigma}
 \]
 Like before, we omit the subscript if clear from the context.
 We also define $(\after{s}{\sigma}) = \after[s]{e_s^0}{\sigma}$.
\end{definition}



\begin{example}
  \label{exa:after-set}
  Consider $s_B$ in Figure~\ref{fig:ALTS-coffee}.
  We evaluate $\after{s_B}{\text{?on}\,\text{?b}\,\text{!t}}$, as follows:
  \begin{align*}
     & \after{q_B^0}{\text{?on}\,\text{?b}\,\text{!t}}
       = \after{q_B^0[\subst{\text{?on}}]}{\text{?b}\,\text{!t}}
       = \after{T(q_B^0,\text{?on})}{\text{?b}\,\text{!t}}\\
    =\;& \after{(q_B^1 \wedge q_B^3 \wedge q_B^5 \wedge q_B^8)}{ \text{?b !t}} = \after{(q_B^1 \wedge q_B^3 \wedge q_B^5 \wedge q_B^8)[\subst{\text{?b}}]}{ \text{!t}} \\
    =\;& \after{(\top \wedge q_B^4 \wedge (q_B^6 \vee q_B^7) \wedge q_B^9)}{\text{!t}}
       = (q_B^4 \wedge (q_B^6 \vee q_B^7) \wedge q_B^9)[\subst{\text{!t}}]\\
    =\;& (\top \wedge (\bot \vee \bot) \wedge q_B^{10})
    = \bot
  \end{align*}
  Intuitively, this means that giving a tea without milk after $\text{?on}\,\text{?b}$ is forbidden.
  In contrast, tea with milk is allowed, and leads to configuration $q_B^{10}$:
  \begin{align*}
     & \after{q_B^0}{\text{?on}\,\text{?b}\,\text{!t+m}}
      = (q_B^4 \wedge (q_B^6 \vee q_B^7) \wedge q_B^9)[\subst{\text{!t+m}}]
      = \top \wedge (\bot \vee \top) \wedge q_B^{10}
      = q_B^{10}
  \end{align*}
\end{example}


\begin{longversion}
  Before we define the semantics of AIA, we establish some essential properties of the $\after{}{}$-function, in Lemmas~\ref{lem:after-distributes-over-expressions} and~\ref{lem:after-distributes-over-traces}.
\end{longversion}

\begin{lemma}
  \label{lem:after-distributes-over-expressions}
  Let $s \in \alts$, $\sigma \in L^*$ and $e_1, e_2 \in \alts$.
  Then
  \begin{align*}
    \after{(e_1 \vee e_2)}{\sigma} &= (\after{e_1}{\sigma}) \vee (\after{e_2}{\sigma})\\
    \after{(e_1 \wedge e_2)}{\sigma} &= (\after{e_1}{\sigma}) \wedge (\after{e_2}{\sigma})
  \end{align*}
\end{lemma}

\begin{proof}
  Any substitution $[f]$ distributes over $\vee$ and $\wedge$ by Definition~\ref{def:substitution} of substitution.
  Since the after-function for fixed $\sigma$ is a successive application of substitutions, this also distributes over $\vee$ and $\wedge$ (formally proven by induction on the length of $\sigma$).
  \qed
\end{proof}

\begin{lemma}
  \label{lem:after-distributes-over-traces}
  Let $s \in \alts$, $e \in \SE(Q_s)$ and $\sigma_1, \sigma_2 \in L^*$.
  Then \[\after{e}{(\sigma_1 \sigma_2)} = \after{(\after{e}{\sigma_1})}{\sigma_2}.\]
\end{lemma}

\begin{proof}
  By induction on the length of $\sigma_1$.
  For base case $\sigma_1 = \epsilon$, we have
  \[\after{e}{\sigma_1 \sigma_2}
    = \after{(\after{e}{\epsilon})}{\sigma_2}
    = \after{e}{(\epsilon \cdot \sigma_2)}
    = \after{e}{(\sigma_1 \cdot \sigma_2)}
  .\]
  For the inductive case, let $\sigma_1 = \ell\sigma_1'$ for some $\ell$ and $\sigma_1'$, and assume as inductive hypothesis (IH) that $\after{e}{(\sigma_1'\sigma_2)} = \after{(\after{e}{\sigma_1'})}{\sigma_2}$ for all state expressions $e$. Then
  \begin{align*}
      & \after{e}{\sigma_1 \sigma_2}
    = \after{e}{\ell\sigma_1'\sigma_2}
    = \after{e[T_\ell]}{\sigma_1'\sigma_2}\\
    =\;& \after{(\after{e[T_\ell]}{\sigma_1'})}{\sigma_2} \ttag{assumption (IH)}\\
    =\;& \after{(\after{e}{\ell\sigma_1'})}{\sigma_2}
    = \after{(\after{e}{\sigma_1})}{\sigma_2} \qedtag
  \end{align*}
\end{proof}


%% file: alts-ir.tex
\subsection{Input-Failure Semantics for AIAs}
\label{sec:alts-ir}

IAs are equipped with input-failure semantics, based on the traces and underspecified inputs of the IA.
We lift this to AIAs via the $\after{}{}$-function,
using that $\bot$ indicates forbidden behaviour, and $\top$ indicates underspecified behaviour.

\begin{definition}
  \label{def:atraces-alts}
  \todo{not too happy with the prime of $s'$}
  Let $s, s' \in \alts$, and $e \in \SE(Q_s)$.
  Then we define
  \begin{align*}
    \irtraces_{s}(e) &\isdef \{\sigma \in L^* \mid (\after[s]{e}{\sigma}) \neq \bot\} \cup \{\sigma \overline a \in L^* \cdot \overline{I} \mid (\after[s]{e}{\sigma a}) = \top\}\\
    \irtraces(s) &\isdef \irtraces_{s}(e_s^0)\\
    s \ir s' &\ifdef \irtraces(s) \subseteq \irtraces(s')\\
    s \ireq s' &\ifdef \irtraces(s) = \irtraces(s')
  \end{align*}
\end{definition}

Compare Definition~\ref{def:ir-inclusion} and Definition~\ref{def:atraces-alts} for input-failure refinement for IAs and for AIAs.
For AIAs, refinement is defined directly over their $\irtraces$, whereas for IA, the input-failure closure of the $\irtraces$ is used for the right-hand model (and optionally for the left-hand model, according to Proposition~\ref{pro:rcl-irtraces-canonical}).
In this regard, AIAs are a more direct and natural representation of input-failure traces, since the input-failure closure is not needed.

\begin{proposition}
  \label{pro:aia-is-ir-closed}
  For $s \in \alts$, $\atraces(s)$ is input-failure closed.
\end{proposition}

\begin{proof}
  To prove input-refusal closedness of $\atraces(s)$, we follow Definition~\ref{def:ir-inclusion} and assume some $\sigma \in L$ and $a \in I_s$ with $\sigma \overline{a} \in \atraces(s)$, and $\rho \in \iftracesdom_{I_s,O_s}$, for which we prove $\sigma a \rho \in \atraces(s)$.
  From $\sigma \overline{a} \in \atraces(s)$ and Definition~\ref{def:atraces-alts}, we find $(\after{s}{\sigma}) = \top$.
  We now distinguish two cases:
  \begin{itemize}
    \item 
      If $\rho \in L_s^*$, then $(\after{s}{\sigma}) = \top$ implies $(\after{s}{\sigma a \rho}) = \top \neq \bot$, which implies $\sigma a \rho \in \irtraces(s)$ by Definition~\ref{def:atraces-alts}.
    \item
      If $\rho = \rho'\overline{b}$, then $(\after{s}{\sigma}) = \top$ implies $(\after{s}{\sigma a \rho' b}) = \top$, which implies $\sigma a \rho = \sigma a \rho' \overline{b} \in \atraces(s)$ by Definition~\ref{def:atraces-alts}.
  \end{itemize}
  So indeed, $\sigma a \rho \in \irtraces(s)$ holds in both cases, proving the proposition.
  \qed
\end{proof}

Another motivation to represent input-failure traces with AIAs is the connection between the distributive lattice $\mathcal{D}(Q)$ and the lattice of sets of input-failure traces: $\wedge$ and $\vee$ are connected to intersection and union of input-failure traces, respectively, and $\top$ and $\bot$ represent the largest and smallest possible input-failure trace sets.

\begin{proposition}
  \label{pro:order-isomorphism}
  
  Let $s \in \alts$, and $e,e' \in \SE(Q_s)$. Then
  \begin{enumerate}
    \item \label{pro:order-isomorphism-wedge}
      $\atraces(e \wedge e') = \atraces(e) \cap \atraces(e')$
    \item \label{pro:order-isomorphism-vee}
      $\atraces(e \vee e') = \atraces(e) \cup \atraces(e')$
    \item \label{pro:order-isomorphism-F}
      $\atraces(\bot) = \emptyset$
    \item \label{pro:order-isomorphism-T}
      $\atraces(\top) = \iftracesdom_{I,O}$
    \item \label{pro:order-isomorphism-trans}
      $\atraces(e) = \{\epsilon\} \cup \{\overline{a} \in \overline{I_s} \mid \after{e}{a} = \top \} \\\strut\hspace{26mm}\cup (\bigcup_{\ell \in L_s} \ell \cdot \atraces(\after{e}{\ell})) \text{ \qquad if $e \neq \bot$}$
  \end{enumerate}
\end{proposition}

\begin{proof}
  \begin{enumerate}
    \item
      We prove
      $\sigma \in \atraces(e \wedge e') \iff \sigma \in (\atraces(e) \cap \atraces(e'))$ for all $\sigma \in \iftracesdom_{I,O}$,
      as follows:
      
      \begin{align*}
	    & \sigma \in \atraces(e \wedge e')\\
	\iff&
	    (\sigma \in L^* \text{ and } (\after{(e \wedge e')}{\sigma}) \neq \bot)\\
	    & \;\text{ or } (\sigma = \sigma'\overline{a} \text{ and } (\after{(e \wedge e')}{\sigma'a}) = \top)
        \ttag{Definition~\ref{def:atraces-alts}}\\
	\iff&
	    (\sigma \in L^* \text{ and } ((\after{e}{\sigma}) \wedge (\after{e'}{\sigma})) \neq \bot)\\
	    & \;\text{ or } (\sigma = \sigma'\overline{a} \text{ and } ((\after{e}{\sigma'a}) \wedge (\after{e'}{\sigma'a})) = \top)
        \ttag{Lemma~\ref{lem:after-distributes-over-expressions}}\\
	\iff&
	    (\sigma \in L^* \text{ and } ((\after{e}{\sigma}) \neq \bot \text{ and } (\after{e'}{\sigma}) \neq \bot))\\
	    & \;\text{ or } (\sigma = \sigma'\overline{a} \text{ and } ((\after{e}{\sigma'a}) = \top \text{ and } (\after{e'}{\sigma'a}) = \top))
        \ttag{\hspace{0.2mm}$e_1 \wedge e_2 = \bot \text{ if and only if } e_1 = \bot \text{ or } e_2 = \bot$}\\
        \ttag{\hspace{0.2mm}$e_1 \wedge e_2 = \top \text{ if and only if } e_1 = \top \text{ and } e_2 = \top$}\\
        &\ttag{by the axioms in Definition~\ref{def:free-distributive-lattice}}\\
	\iff&
	    (\;(\sigma \in L^* \text{ and } (\after{e}{\sigma}) \neq \bot)\\
	    & \quad\text{ or } (\sigma = \sigma'\overline{a} \text{ and } (\after{e}{\sigma'a}) = \top))\\
	    &\text{and}\\
	    &(\;(\sigma \in L^* \text{ and } (\after{e'}{\sigma}) \neq \bot)\\
	    & \quad\text{ or } (\sigma = \sigma'\overline{a} \text{ and } (\after{e'}{\sigma'a}) = \top))
        \ttag{distributivity of logical conjunction and disjunction}\\
	\iff&
	    \sigma \in \atraces(e) \text{ and } \sigma \in \atraces(e') 
        \ttag{Definition~\ref{def:atraces-alts}}\\
	\iff&
	    \sigma \in \atraces(e) \cap \atraces(e') \ttag{set theory}
      \end{align*}
    \item
      Proof is analogous to the proof for $e \wedge e'$.
    \item
      As $(\after{\bot}{\sigma}) = \bot$ for any $\sigma \in L^*$, this is trivial from Definition~\ref{def:atraces-alts}.
    \item
      As $(\after{\top}{\sigma}) = \top$ for any $\sigma \in L^*$, this is trivial from Definition~\ref{def:atraces-alts}.
    \item
      Assume $e \neq \bot$. Then
      \begin{align*}
         &
	  \atraces(e) \\
        =&
	  \{\sigma \in L^* \mid (\after{e}{\sigma}) \neq \bot\} \cup \{\sigma \overline a \in L^* \cdot \overline{I} \mid (\after{e}{\sigma a}) = \top\} \ttag{Definition~\ref{def:atraces-alts}}\\
        =&
	  \{\sigma \in \{\epsilon\} \mid (\after{e}{\sigma}) \neq \bot\} \\
	 &\cup
	  \{\sigma \in \bigcup_{\ell \in L}(\ell \cdot L^*) \mid (\after{e}{\sigma}) \neq \bot\}\\
	 &\cup
	  \{\sigma \overline a \in \overline{I} \mid (\after{e}{\sigma a}) = \top\}\\
	 &\cup
	  \{\sigma \overline a \in \bigcup_{\ell \in L}(\ell \cdot L^* \cdot \overline{I}) \mid (\after{e}{\sigma a}) = \top\}
	    \ttag{$L^* = \{\epsilon\} \cup \bigcup_{\ell \in L} (\ell \cdot L^*)$ and $L^* \cdot \overline{I} = \overline{I} \cup \bigcup_{\ell \in L} (\ell \cdot L^* \cdot \overline{I})$}\\
	=&
	  \{\epsilon \mid (\after{e}{\epsilon}) \neq \bot\} \\
	 &\cup
	  \bigcup_{\ell \in L} (\ell \cdot \{\sigma \in L^* \mid (\after{e}{\ell\sigma}) \neq \bot\})\\
	 &\cup
	  \{\overline a \in \overline{I} \mid (\after{e}{a}) = \top\}\\
	 &\cup
	  \bigcup_{\ell \in L} (\ell \cdot \{\sigma \overline a \in (L^* \cdot \overline{I}) \mid (\after{e}{\ell \sigma a}) = \top\})
	    \ttag{rewriting}\\
	=&
	  \{\epsilon \mid e \neq \bot\} \cup
	  \{\overline a \in \overline{I} \mid (\after{e}{a}) = \top\}\\
	 &\cup
	  \bigcup_{\ell \in L} \ell \cdot (\{\sigma \in L^* \mid (\after{(\after{e}{\ell})}{\sigma}) \neq \bot\} \\[-3mm]
	  &\hspace{15mm}\cup \{\sigma \overline a \in L^* \cdot \overline{I} \mid (\after{(\after{e}{\ell})}{\sigma a}) = \top\})
	    \ttag{Lemma~\ref{lem:after-distributes-over-traces}, and $(\after{e}{\epsilon}) = e$}\\
	=&
	  \{\epsilon\} \cup
	  \{\overline a \in \overline{I} \mid (\after{e}{a}) = \top\}
	  \cup \bigcup_{\ell \in L} \ell \cdot (\atraces(\after{e}{\ell}))
	    \ttag{Definition~\ref{def:atraces-alts}, and the assumption $e \neq \bot$}
      \end{align*}
  \end{enumerate}
  \strut\qed
\end{proof}

Propositions~\ref{pro:order-isomorphism}.\ref{pro:order-isomorphism-F} and \ref{pro:order-isomorphism}.\ref{pro:order-isomorphism-trans} show why Definition~\ref{def:alts} does not allow transitions to $T(q,a) = \bot$ for an input $a$: in that case, $\irtraces(q)$ would contain trace $\epsilon$, but it would not contain extension $a$ nor $\overline{a}$ of $\epsilon$, meaning that after trace $\epsilon$ it is not allowed to accept nor to refuse $a$.

We can lift configurations $\top$ and $\bot$, as well as $\wedge$ and $\vee$, to the level of AIAs.
This provides the building blocks to compose specifications.
Specifications $s_\top$ and $s_\bot$ can be used to specify that any or no behaviour is considered correct, respectively.
The operators $\wedge$ and $\vee$ on specifications fulfill the same role as existing operators in substitutivity refinement~\cite{Chietal14}, and have similar properties, described in Proposition~\ref{pro:order-isomorphism}.

\begin{definition}
  Let $s_1, s_2 \in \alts$.
  Without loss of generality\footnote{
  If $Q_1$ and $Q_2$ are not disjoint, the disjoint union $Q_1 \uplus Q_2$ can be used instead of $Q_1 \cup Q_2$.
  The transition functions of $s_1 \wedge s_2$ and $s_1 \vee s_2$ should be adjusted accordingly.
  },
  assume that $Q_1$ and $Q_2$ are disjoint.
  We define 
  \begin{align*}
    s_\top &\isdef (\emptyset, I, O, \emptyset, \top)&
    s_1 \wedge s_2 &\isdef (Q_1 \cup Q_2, I, O, T_1 \cup T_2, e_1^0 \wedge e_2^0)\\
    s_\bot &\isdef (\emptyset, I, O, \emptyset, \bot)&
    s_1 \vee s_2 &\isdef (Q_1 \cup Q_2, I, O, T_1 \cup T_2, e_1^0 \vee e_2^0)
  \end{align*}
\end{definition}

\begin{proposition}
  \label{pro:order-isomorphism-aia}
  Let $i, i', s, s'\in \alts$.
  Then
  \begin{align*}
    i \ir s \text{ and } i \ir s' &\iff i \ir (s \wedge s')\\
    i \ir s \text{ or } i \ir s' &\implies i \ir (s \vee s')\\
    i \ir s \text{ and } i' \ir s &\iff (i \vee i') \ir s\\
    i \ir s \text{ or } i' \ir s &\implies (i \wedge i') \ir s\\
    i &\ir s_\top\\
    i &\not\ir s_\bot \text{ \qquad if $e_i^0 \neq \bot$}
  \end{align*}
\end{proposition}

\begin{proof}
  We prove the first statement:
  \begin{align*}
    \iff & i \ir s \text{ and } i \ir s' \\
    \iff & \irtraces(i) \subseteq \irtraces(s) \text{ and } \irtraces(i) \subseteq \irtraces(s') \ttag{Definition~\ref{def:atraces-alts}}\\
    \iff & \irtraces(i) \subseteq \irtraces(s) \cap \irtraces(s') \ttag{basic set theory}\\
    \iff & \irtraces(i) \subseteq \irtraces(s \wedge s') \ttag{Proposition~\ref{pro:order-isomorphism}}\\
    \iff & i \ir (s \wedge s') \ttag{Definition~\ref{def:atraces-alts}}
  \end{align*}
  The other statements can be derived analogously.
  \qed
\end{proof}

The converse of statement (2) does not hold: if $\irtraces(i) = \{\epsilon,x,y\}$, $\irtraces(s_1) = \{\epsilon,x\}$ and $\irtraces(s_2) = \{\epsilon,y\}$, then $i \ir s_1 \vee s_2$ holds, but  $i \not\ir s_1$ and $i \not\ir s_2$.
The converse of statement (4) can be disproven similarly.

%% file: alts-determinization.tex
\subsection{AIA Determinization}
\label{sec:determinization}

In case of nestings of $\wedge$ and $\vee$, the after-set $\after{s}{\sigma}$ may not be clear immediately, so a transition function producing configurations without $\wedge$ and $\vee$ is easier to interpret.
For this reason, we lift the notions of determinism and determinization from IAs~\cite{JaVaTr19} to the alternating setting.

\begin{definition}
  \label{def:alts-deterministic}
  Let $s \in \alts$ and $e \in \SE(Q_s)$.
  Then $e$ is \emph{deterministic} if $e = \top$ or $e = \bot$ or $e = \AP{q}$ for some $q \in Q_s$.
  Furthermore, $s$ is deterministic if for all $\sigma \in L^*$, configuration $\after{s}{\sigma}$ is deterministic.
\end{definition}

Compare the notions of determinism for IAs and AIAs.
For every trace $\sigma$, a deterministic IA $s$ is in a singleton state $(\after{s}{\sigma}) =\{q\}$, unless $(\after{s}{\sigma}) = \emptyset$ (that is, $\sigma$ is not a trace of $s$).
For AIAs, this singleton set $\{q\}$ is replaced by the embedding $\AP{q}$, and $\emptyset$ is replaced by $\top$ or $\bot$, depending on whether this set was reached by an undespecified action or a forbidden action.

We now define determinization, where we require the distinction between $\AP{q}$ and $q$ to avoid ambiguity.

\begin{definition}
  \label{def:alts-determinization}
  Let $s \in \alts$.
  We define $\det : \SE(Q_s) \rightarrow \SE(\SE(Q_s) \setminus \{\top,\bot\})$ as
  \begin{align*}
    \det(e) \isdef& \begin{cases}
      \top &\text{ if $e = \top$}\\
      \bot &\text{ if $e = \bot$}\\
      \AP{e} &\text{ otherwise}
    \end{cases}
  \end{align*}
  The \emph{determinization of $s$}, or $\det(s) \in \alts$, is defined as
  \begin{align*}
    \det(s) &\isdef (\SE(Q_s) \setminus \{\top,\bot\},I,O,T_{\det(s)},\det(e_s^0))\text{, with}\\
    T_{\det(s)}(e,\ell) &\isdef \det(\after[s]{e}{\ell}) \text{\qquad for $\ell \in L$}\\
  \end{align*}
\end{definition}

\begin{proposition}
  \label{def:alts-det-is-deterministic}
  For $s \in \alts$, $\det(s)$ is deterministic.
\end{proposition}

\begin{proof}
  By Definition~\ref{def:alts-deterministic}, $\det(s)$ is deterministic if $\after{s}{\sigma}$ is deterministic for all $\sigma \in L^*$.
  We prove this by induction to the length of $\sigma$:
  \begin{itemize}
    \item
      For $\sigma = \epsilon$, we have $(\after{s}{\sigma}) = e_{\det(s)}^0 = \det(e_s^0)$, which is deterministic by Definition~\ref{def:alts-determinization}.
    \item
      For $\sigma = \sigma'\ell$, assume as induction hypothesis that $\after{s}{\sigma'}$ is deterministic.
      We distinguish two cases:
      \begin{itemize}
        \item
        If $\after{\det(s)}{\sigma'}$ is $\top$ or $\bot$, then $\after{\det(s)}{\sigma}$ is respectively $\top$ or $\bot$ as well, by Definition~\ref{def:sla-after}, and in particular it is deterministic.
        \item
        If $(\after{\det(s)}{\sigma'}) = \AP{e}$ for $e \in \SE(Q_s) \setminus \{\top,\bot\}$, then 
        \begin{align*}
          & (\after{\det(s)}{\sigma})\\
          =& \after{(\after{\det(s)}{\sigma'})}{\ell} \ttag{Lemma~\ref{lem:after-distributes-over-traces}}\\
          =& \after[\det(s)]{\AP{e}}{\ell} \ttag{case assumption}\\
          =& \AP{e}[{T_{\det(s)}}_\ell] \ttag{Definition~\ref{def:sla-after}}\\
          =& T_{\det(s)}(e,\ell) \ttag{Definition~\ref{def:substitution}}\\
          =& \det(\after[s]{e}{\ell}) \ttag{Definition~\ref{def:alts-determinization}}
        \end{align*}
        which is deterministic.
      \end{itemize}
  \end{itemize}
  Thus, $\after{\det(s)}{\sigma}$ is always deterministic, so $\det(s)$ is deterministic as well, by Definition~\ref{def:alts-deterministic} of determinism.
  \qed
\end{proof}

\begin{example}
  \label{exa:determinization}
  Figure~\ref{fig:example-determinizations} shows (the reachable part of) the determinizations of $s_A$ and $s_B$ from Figure~\ref{fig:alts-example}.
  In $\det(s_A)$, state $q_A^0 \wedge q_A^2$ has no outgoing !x-transition.
  This expresses $T_{\det(s_A)}(q_A^0 \wedge q_A^2, \text{!x}) = \bot$, which is because $q_A^2$ has no $x$-transition, $T_A(q_A^0, \text{!x}) = \bot$.
  In contrast, state $q_A^0 \wedge q_A^2$ has an outgoing ?a-transition, $T_{\det(s_A)}(q_A^0 \wedge q_A^2, \text{?a}) \neq \top$, because $q_A^0$ has an ?a-transition, $T_A(q_A^0, \text{?a}) \neq \top$.
  
  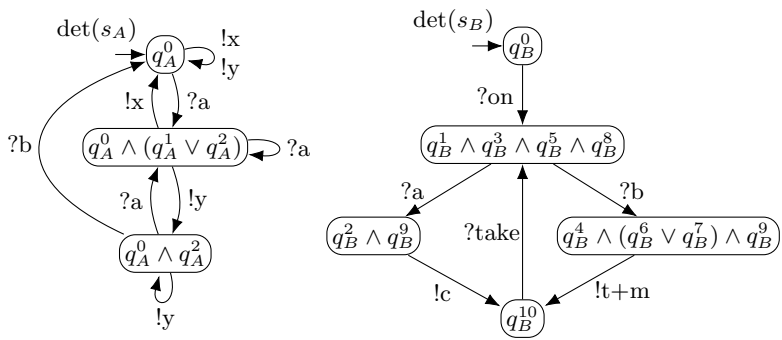
\begin{figure}
    \begin{center}
      \tikzlts
      \begin{tikzpicture}[node distance=20mm]
        \node [recstate] [initial,initial text=] (0) {${q_A^0}$};
        \node (init-text) [above left=-2mm and 0 of 0] {$\det(s_A)$};
        \node [recstate] (0_12) [below=7.2mm of 0] {${q_A^0} \wedge ({q_A^1} \vee {q_A^2})$};
        \node [recstate] (0_2) [below=9mm of 0_12] {${q_A^0} \wedge {q_A^2}$};
        \path
    (0) edge [bend left=20] node [right] {?a} (0_12)
    (0) edge [loop right] node [align=left] {!x\\!y} (0)
    (0_12) edge [loop right,in=-5,out=5,looseness=7] node {?a} (0_12)
    (0_12) edge [bend left=20] node [left] {!x} (0)
    (0_12) edge [bend left=20] node [right] {!y} (0_2)
    (0_2) edge [bend left=20] node [left] {?a} (0_12)
    (0_2) edge [loop below] node {!y} (0_2)
    (0_2) edge [in=200,out=155,looseness=2] node [left] {?b} (0)
        ;
      \end{tikzpicture}
      \begin{tikzpicture}[->,node distance=18mm]
        \tikzlts
        \node[initial left,recstate] (init) {${q_B^0}$};
        \node (init-text) [above left=-2mm and 0 of init] {$\det(s_B)$};
        \node[recstate] (on) [below=8mm of init] {${q_B^{1}} \wedge {q_B^3} \wedge {q_B^5} \wedge {q_B^8}$};
        \node[recstate] (A) [below left=7mm and 0mm of on] {${q_B^{2}} \wedge {q_B^9}$};
        \node[recstate] (B) [below right=7mm and -9mm of on, align=left] {${q_B^4} \wedge ({q_B^6} \vee {q_B^7}) \wedge q_B^9$};
        \node[recstate] (take) [below=18mm of on] {${q_B^{10}}$};
        \path
        (init) edge node [left] {?on} (on)
        (on) edge node [left=2mm] {?a} (A)
        (on) edge node [right=2mm] {?b} (B)
        (A) edge node [below left=-1mm] {!c} (take)
        (B) edge node [below right=-1mm and -0.5mm] {!t+m} (take)
        (take) edge node [left=-0.8mm] {?take} (on)
        ;
      \end{tikzpicture}
      \caption{
      Examples of determinization.
      }
      \label{fig:example-determinizations}
    \end{center}
  \end{figure}
\end{example}

Example~\ref{exa:determinization} shows that an input is specified by a conjunction of states in the determinization if \emph{any} of the individual state specify this input, whereas an output is allowed by a conjunction of states only if \emph{all} of the individual state allow this output.
In the setting of IA, \cite{JaVaTr19} already established that this works in a reversed way for non-determinism, following their definition of determinization: \emph{all} individual states of a disjunction should specify an input to specify it in the determinization, and \emph{any} individual state should allow an output to allow it in the determinization.
Their so-called \emph{input-universal determinization} is an instance of the determinization from Definition~\ref{def:alts-determinization}, using only disjunctions.

This duality arises from Definition~\ref{def:sla-after} of $\after{}{}$, since the determinization directly represents the $\after{}{}$-function: the determinizations in Example~\ref{exa:determinization} correspond to the $\after{}{}$-sets such as those derived in Example~\ref{exa:after-set}.
This correspondence is formalized in Proposition~\ref{pro:det-after}.
  
\begin{proposition}
  \label{pro:det-after}
  Let $s \in \alts$ and $\sigma \in L^*$.
  Then
  \[(\after{\det(s)}{\sigma}) = \det(\after{s}{\sigma}).\]
\end{proposition}

\begin{proof}
  For $e \in \SE(Q_s)$, we first prove 
  \[(\after[\det(s)]{\det(e)}{\sigma}) = \det(\after[s]{e}{\sigma})\]
  by induction on the length of $\sigma$.
  \begin{itemize}
    \item
      If $\sigma = \epsilon$, then 
      $(\after[\det(s)]{\det(e)}{\sigma})
        = \det(e)
        = \det(\after[s]{e}{\sigma})$,
      by Definition~\ref{def:sla-after}.
    \item
      If $\sigma = \sigma'\ell$, assume as induction hypothesis $H_1$ that the lemma holds for $\sigma'$.
      Then we distinguish the following cases:
      \begin{itemize}
        \item
          If $(\after[s]{e}{\sigma'}) = \top$, then:
          \begin{align*}
             & (\after[\det(s)]{\det(e)}{\sigma})\\
            =& \after[\det(s)]{(\after[\det(s)]{\det(e)}{\sigma'})}{\ell} \ttag{Definition~\ref{def:sla-after} and $\sigma=\sigma'\ell$}\\
            =& \after[\det(s)]{(\det(\after[s]{e}{\sigma'}))}{\ell} \ttag{$H_1$}\\
            =& \after[\det(s)]{\det(\top)}{\ell} \ttag{case assumption}\\
            =& \after[\det(s)]{\top}{\ell} \ttag{Definition~\ref{def:alts-determinization} of $\det$}\\
            =& \top \ttag{Definition~\ref{def:sla-after}}\\
            =& \det(\top) \ttag{Definition~\ref{def:alts-determinization} of $\det$}\\
            =& \det(\after[s]{\top}{\ell}) \ttag{Definition~\ref{def:sla-after}}\\
            =& \det(\after[s]{(\after[s]{e}{\sigma'})}{\ell}) \ttag{case assumption}\\
            =& \det(\after[s]{e}{\sigma}) \ttag{Definition~\ref{def:sla-after} and $\sigma = \sigma'\ell$}
          \end{align*}
        \item If $(\after[s]{e}{\sigma'}) = \bot$, then the lemma holds analogously to the previous case.
        \item If $(\after[s]{e}{\sigma'}) \neq \top$ and $(\after[s]{e}{\sigma'}) \neq \bot$, then:
        \begin{align*}
      & \after[\det(s)]{\det(e)}{\sigma} \\
    =& \after[\det(s)]{(\after[\det(s)]{\det(e)}{\sigma'})}{\ell} \ttag{Definition~\ref{def:sla-after} and $\sigma = \sigma'\ell$}\\
    =& \after[\det(s)]{\det(\after[s]{e}{\sigma'})}{\ell} \ttag{$H_1$}\\
    =& \after[\det(s)]{\AP{\after[s]{e}{\sigma'}}}{\ell} \ttag{Definition~\ref{def:alts-determinization} of $\det$ and case assumption}\\
    =& T_{\det(s)}(\after[s]{e}{\sigma'}, \ell) \ttag{Definition~\ref{def:sla-after}}\\
    =& \det(\after[s]{(\after[s]{e}{\sigma'})}{\ell}) \ttag{Definition~\ref{def:alts-determinization} of $T_{\det(s)}$}\\
    =& \det(\after[s]{e}{\sigma}) \ttag{Definition~\ref{def:sla-after} and $\sigma = \sigma'\ell$}
  \end{align*}
      \end{itemize}
  \end{itemize}
  The proposition then follows:
  \begin{align*}
     & \after{\det(s)}{\sigma}\\
    =& \after{e_{\det(s)}^0}{\sigma} \ttag{Definition~\ref{def:sla-after}}\\
    =& \after{\det(e_s^0)}{\sigma} \ttag{Definition~\ref{def:alts-determinization}}\\
    =& \det(\after{e_s^0}{\sigma}) \qedtag
  \end{align*}
\end{proof}

\begin{proposition}
  \label{pro:determinization-preserves-irtraces}
  Let $s \in \alts$. Then
  $\atraces(s) = \atraces(\det(s))$.
\end{proposition}

\begin{proof}
  In Definition~\ref{def:atraces-alts}, we observe that whether a trace $\sigma \in \mathcal{AT}_{I_s,O_s}$ is an alternating trace of $s$ depends only on whether $(\after{s}{\sigma})$ is $\top$ or $\bot$.
  From Lemma~\ref{pro:det-after}, we know that this property is preserved by determinization:   \begin{align*}
    (\after{s}{\sigma}) = \top &\iff (\after{\det(s)}{\sigma}) = \top\\
    (\after{s}{\sigma}) = \bot &\iff (\after{\det(s)}{\sigma}) = \bot
  \end{align*}
  This implies that the alternating traces are also preserved.
  \qed
\end{proof}

\begin{corollary}
  Let $s \in \alts$. Then $s \ireq \det(s)$.
\end{corollary}

A known result~\cite{ChKoSt81} is that alternating automata are exponentially more succinct than non-deterministic automata, and double exponentially more succinct than deterministic automata.
Although alternating automata are not a special case of AIAs (as AIAs lack the accepting and non-accepting states of alternating automata), we expect AIAs to be exponentially more succinct than IAs, as well.

%% file: alts-lts.tex
\subsection{Connections between IAs and AIAs}
\label{sec:alts-lts}

IAs and AIAs are used to represent sets of input-failure traces, and are in that sense interchangeable.
First, we show that any IA can be translated to an AIA.

\begin{definition}
 \label{def:alts-induced}
 For $s \in \lts$, the \emph{AIA induced by $s$} is defined as $\altsop(s) \isdef (Q_s,I_s, O_s,T,\bigvee Q_s^0) \in \alts$, where for all $q \in Q_s$ and $\ell \in L$:
 \[T(q,\ell) \isdef \begin{cases}
        \top &\text{ if $\ell \in I$ and $q \notrel{\xrightarrow{\ell}}$}\\
        \bigvee\, T_s(q,\ell) &\text{ otherwise}
        \end{cases}
 \]
\end{definition}

\begin{longversion}
  Translating an IA to an AIA should preserve input-failure traces. 
  In order to prove this in Proposition~\ref{pro:atraces-lts-alts}, we need some auxiliary lemmas.
\end{longversion}

\begin{lemma}
  \label{lem:aia-induced-transition}
  Let $s \in \lts$, $\ell \in L$ and let $Q \subseteq Q_s$ be finite.
  Then
  \[
  (\bigvee Q)[{T_{\altsop(s)}}_\ell] = 
    \begin{cases}
      \top & \text{if $\ell \in I \setminus \inp_{s}(Q)$}\\
      \bigvee \{q \in Q_s \mid \exists q' \in Q: q' \xrightarrow{\ell}_s q\} & \text{otherwise}
    \end{cases}
  \]
\end{lemma}

\begin{proof}
  We prove the first part of the equality by assuming $\ell \in I \setminus \inp_s(Q)$.
  \begin{align*}
        & \ell \in I \setminus \inp_s(Q)\\
    \implies& \ell \in I \wedge \exists q \in Q: \ell \not\in \inp_s(q)  
      \ttag{Definition~\ref{def:lts-definitions}}\\
    \implies& \exists q \in Q: \ell \in I \setminus \inp_s(q)
      \ttag{Definition~\ref{def:lts-definitions}}\\
    \implies& \exists q \in Q: T_{\altsop(s)}(q,\ell) = \top
      \ttag{Definition~\ref{def:alts-induced}}\\
    \implies& \exists q \in Q: q[{T_{\altsop(s)}}_\ell] = \top
      \ttag{Definition~\ref{def:substitution}}\\
    \implies& \bigvee \{q [{T_{\altsop(s)}}_\ell] \mid q \in Q\} = \top
      \ttag{identity law of Definition~\ref{def:free-distributive-lattice} (formally, by induction on the size of $Q$)}\\
    \implies& \bigvee Q [{T_{\altsop(s)}}_\ell]  = \top
      \ttag{Definition~\ref{def:substitution} (formally, by induction)}
  \end{align*}
  For the second part of the equality, assume $\ell \not\in I \setminus \inp_s(Q)$.
  Then
  \begin{align*}
    &  \bigvee \{q \in Q_s \mid \exists q' \in Q: q' \xrightarrow{\ell}_s q\}\\
    =& \bigvee \{\bigvee \{q' \in Q_s \mid q \xrightarrow{\ell} q'\} \mid q \in Q\}
      \ttag{Laws of $\vee$ in Definition~\ref{def:free-distributive-lattice}}\\
    =& \bigvee \{T_{\altsop(s)}(q,\ell) \mid q \in Q\} 
      \ttag{Definition~\ref{def:alts-induced} and assumption $\ell \not\in I \setminus \inp_s(Q)$}\\
    =& (\bigvee Q)[{T_{\altsop(s)}}_\ell] .
      \ttag{Definition~\ref{def:substitution} (formally, by induction)}
  \end{align*}
  \strut\qed
\end{proof}

\begin{lemma}
  \label{lem:input-universal-induced-aia}
  Let $s \in \lts$, $\sigma = \ell^1 \dots \ell^n \in L^*$, and $Q \subseteq Q_s$.
  If we define that $\sigma$ is called \emph{$s$-universal for $Q$} if 
  \[\forall j \in \{1 \dots n\}: \ell^j \in \inp(\after{Q}{\ell^1\dots\ell^{j-1}}) \cup O\]
  then the following holds:
  \begin{align*}
    (\text{$\sigma$ is $s$-input-universal for $Q_s^0$}) &\iff \after{\altsop(s)}{\sigma} = \bigvee (\after{s}{\sigma})\\
    (\text{$\sigma$ is not $s$-input-universal for $Q_s^0$}) &\iff \after{\altsop(s)}{\sigma} = \top
  \end{align*}
\end{lemma}

\begin{proof}
  We first prove the statements that for finite $Q \subseteq Q_s$, 
  \begin{align*}
    (\text{$\sigma$ is $s$-input-universal for $Q$}) &\implies \after[\altsop(s)]{(\bigvee Q)}{\sigma} = \bigvee (\after[s]{Q}{\sigma})\\
    (\text{$\sigma$ is not $s$-input-universal for $Q$}) &\implies \after[\altsop(s)]{(\bigvee Q)}{\sigma} = \top
  \end{align*}
  holds, by induction on the length of $\sigma$.
  The base case $\sigma = \epsilon$ is vacuously $s$-input universal, so the implication follows directly from Definition~\ref{def:alts-induced}.
  For the inductive step, let $\sigma = \ell\sigma'$ with $\ell \in L$ and assume that the statement holds for $\sigma'$ (IH).
  First, we show that (1):
  \begin{align*}
      & \after[\altsop(s)]{(\bigvee Q)}{\sigma}\\
    =& \after[\altsop(s)]{(\bigvee Q)}{\ell\sigma'}
      \ttag{$\sigma = \ell\sigma'$}\\
    =& \after[\altsop(s)]{(\after[\altsop(s)]{(\bigvee Q)}{\ell})}{\sigma'}
      \ttag{Lemma~\ref{lem:after-distributes-over-traces}}\\
    =& \after[\altsop(s)]{(\bigvee Q)[{T_{\altsop(s)}}_\ell]}{\sigma'}
      \ttag{Definition~\ref{def:sla-after}}
  \end{align*}
  Now we distinguish two cases:
  \begin{itemize}
    \item
      If $\ell \in I \setminus \inp_{s}(Q)$, then $\sigma$ is not input-universal (2).
      This implies
      \begin{align*}
         & \after[\altsop(s)]{(\bigvee Q)}{\sigma}\\
        =& \after[\altsop(s)]{(\bigvee Q)[{T_{\altsop(s)}}_\ell]}{\sigma'}
          \ttag{Observation (1)}\\
        =& \after[\altsop(s)]{\top}{\sigma'}
          \ttag{Lemma~\ref{lem:aia-induced-transition}}\\
        =& \top \qquad(3)
      \end{align*}
      Combining observations (2) and (3), we find that the lemma holds for the inductive step.
    \item
      If $\sigma'$ is not input-universal, then $\sigma$ is also not input-universal by the definition of input-universality (4).
      By (IH), we then have that $\after[\altsop(s)]{(\bigvee Q)}{\sigma'} = \top$, and therefore also $\after[\altsop(s)]{(\bigvee Q)}{\sigma} = \top$ (5).
      Together, (4) and (5) imply that the lemma holds for the inductive step.
    \item
      If $\ell \not\in I\setminus_s(Q)$ and $\sigma'$ is input-universal, then we first observe that $\sigma'$ is also input-universal (6).
      Moreover,
      \begin{align*}
         & \after[\altsop(s)]{(\bigvee Q)}{\sigma}\\
        =& \after[\altsop(s)]{(\bigvee Q)[{T_{\altsop(s)}}_\ell]}{\sigma'}
          \ttag{Observation (1)}\\
        =& \after[\altsop(s)]{\bigvee \{q \in Q_s \mid \exists q' \in Q: q' \xrightarrow{\ell}_s q\}}{\sigma'}
          \ttag{Lemma~\ref{lem:aia-induced-transition}}\\
        =& \after[\altsop(s)]{(\bigvee \after[s]{Q}{\ell})}{\sigma'}
          \ttag{Definition~\ref{def:lts-definitions}}\\
        =& \bigvee (\after[s]{(\after[s]{Q}{\ell})}{\sigma'})
          \ttag{(IH)}\\
        =& \bigvee (\after[s]{Q}{\sigma}) \qquad(7)
      \end{align*}
      Combining (6) and (7), we again find that the statements holds for the inductive step.
  \end{itemize}
  
  To conclude the lemma, we first observe that the proven statements directly imply the statements in the lemma in one direction $(\implies)$ (8).
  The other direction can be proven by contradiction: assume that $\after{\alts(s)}{\sigma} = \top$ (9), then $\after{\alts(s)}{\sigma} \neq \bigvee(\after{s}{\sigma})$ (10) holds, since $\bigvee(\after{s}{\sigma}) \neq \top$.
  Therefore $\sigma$ cannot be $s$-input-universal for $Q_s^0$, as otherwise (8) and (10) would contradict.
  Consequently, we find that (9) implies that $\sigma$ is not input-universal for $Q_s^0$.
  Analogously, we can prove that $\after{\alts(s)}{\sigma} = \bigvee(\after{s}{\sigma})$ implies that $\sigma$ is input-universal for $Q_s^0$, proving the lemma.
  \qed
\end{proof}

\begin{proposition}
  \label{pro:atraces-lts-alts}
  Let $s \in \lts$. Then
  $\irtraces(\altsop(s)) = \rcl(\irtraces(s))$.
\end{proposition}

\begin{proof}
  We prove $\sigma \in \irtraces(\altsop(s)) \iff \sigma \in \rcl(\irtraces(s))$ for some $\sigma \in \iftracesdom_{I,O}$,  
  We make a case distinction:
  \begin{itemize}
    \item
      If $\sigma = \sigma'\overline{a} \in L^* \cdot \overline{I}$, then
      \begin{align*}
             &
          \sigma \in \irtraces(\alts(s))
            \\
        \iff &
          \after{\alts(s)}{\sigma'a} = \top
            \ttag{Definition~\ref{def:atraces-alts}}\\
        \iff &
          \text{$\sigma'a$ is not $s$-input-universal for $Q_s^0$}
            \ttag{Lemma~\ref{lem:input-universal-induced-aia}}\\
        \iff &
          \text{there is a decomposition $\sigma'a = \rho b \rho'$ with $b \in I$ and $\rho,\rho' \in L^*$ such }\\&\text{that $b \not \in \inp(\after{Q_s^0}{\rho})$}
          \ttag{Definition of input-universality}\\
        \iff &
          \text{there is a decomposition $\sigma'a = \rho b \rho'$ such that $\rho\overline{b} \in \irtraces(s)$}
          \ttag{Definition~\ref{def:ir-inclusion}}\\
        \iff &
          \sigma \in \rcl(\irtraces(s))
          \ttag{Definition~\ref{def:ir-inclusion}}
      \end{align*}
    \item
      If $\sigma \in L^*$, then
      \begin{align*}
             &
          \sigma \in \irtraces(\alts(s))
            \\
        \iff &
          \after{\alts(s)}{\sigma} \neq \bot
          \ttag{Definition~\ref{def:atraces-alts} and case assumption $\sigma \in L^*$}\\
        \iff &
          \after{\alts(s)}{\sigma} = \top\\
          &\text{or }(\after{\alts(s)}{\sigma} \neq \bot \text{ and } \after{\alts(s)}{\sigma} \neq \top)
            \ttag{Basic logic}\\
        \iff &
          \text{$\sigma$ is not $s$-input-universal for $Q_s^0$}\\
          &\text{or }(\after{\alts(s)}{\sigma} \neq \bot \text{ and $\sigma$ is $s$-input-universal for $Q_s^0$})
            \ttag{Lemma~\ref{lem:input-universal-induced-aia}}\\
        \iff &
          \text{$\sigma$ is not $s$-input-universal for $Q_s^0$}\\
          & \text{or } (\after{s}{\sigma} \neq \emptyset \text{ and $\sigma$ is $s$-input-universal for $Q_s^0$})
            \ttag{$\after{\alts(s)}{\sigma} = \bigvee(\after{s}{\sigma})$ by Lemma~\ref{lem:input-universal-induced-aia}, and Definition~\ref{def:free-distributive-lattice} of $\bigvee$}\\
        \iff &
          \text{there is a decomposition $\sigma = \rho a \rho'$ with $a \in I$ and $\rho,\rho' \in L^*$ such }\\&\text{that $a \not \in \inp(\after{Q_s^0}{\rho})$}
          \ttag{Definition of input-universality}\\
          & \text{or } (\sigma \in \traces(s))
            \ttag{$\after{s}{\sigma} = \emptyset \iff \sigma \in \traces(s)$, Lemma~\ref{lem:input-universal-induced-aia}}\\
          & \text{and there is no decomposition $\sigma = \rho a \rho'$ such that $a \not \in \inp(\after{Q_s^0}{\rho})$)}
            \ttag{Definition of input-universality}\\
        \iff &
          \text{there is a decomposition $\sigma = \rho a \rho'$ such that $\rho\overline{a} \in \irtraces(s)$}\\
          &\text{or } (\sigma \in \irtraces(s)
          \ttag{Definition~\ref{def:ir-inclusion} and case assumption $\sigma \in L^*$}\\
          & \text{and there is no decomposition $\sigma = \rho a \rho'$ such that $a \not \in \inp(\after{Q_s^0}{\rho})$)}\\
        \iff &
          \sigma \in \rcl(\irtraces(s))
          \ttag{Definition~\ref{def:ir-inclusion}}
      \end{align*}
  \end{itemize}
  \strut\qed
\end{proof}

\begin{corollary}
  Let $s_1, s_2 \in \lts$.
  Then
  $s_1 \ir s_2 \iff \altsop(s_1) \ir \altsop(s_2)$
\end{corollary}

Definition~\ref{def:lts-induced} formalizes how disjunction in an AIA corresponds to non-determinism in IA.
Specifically, if no transitions are present for some output in an IA, then the transition function of the corresponding AIA gives $\bigvee \emptyset = \bot$ for this output, analogous to the explicit case $\top$ for inputs.
Note that the graphical representation of an IA and that of its induced AIA are the same.

The translation from AIAs to IAs is more involved.
For disjunctions of states $(\after{q}{\ell}) = q_1 \vee q_2$, the translation of Definition~\ref{def:alts-induced} can simply be inverted, but this is not possible for conjunctions.
As such, we represent any configuration by its unique disjunctive normal form.

\begin{definition}
 Let $e \in \mathcal{D}(Q)$.
 Then $\operatorname{DNF}(e)$ is the smallest set in $\mathcal{P}(\mathcal{P}(Q))$ such that 
 $e = \bigvee \{\bigwedge Q' \mid Q' \in \operatorname{DNF}(e)\}$.
\end{definition}

The set $\operatorname{DNF}(e)$ can be constructed by using the axioms from Definition~\ref{def:free-distributive-lattice}.

\begin{example}
To find $\operatorname{DNF}(q^1 \vee (q^2 \wedge (q^1 \vee q^3)))$, we first rewrite the expression by using distributivity, associativity, commutativity and absorbtion, as follows:
 \[
     q^1 \vee (q^2 \wedge (q^1 \vee q^3))
   = q^1 \vee (q^2 \wedge q^1) \vee (q^2 \wedge q^3)
   = q^1 \vee (q^2 \wedge q^3)
 \]
 So we find $\operatorname{DNF}(q^1 \vee (q^2 \wedge (q^1 \vee q^3))) = \{\{q^1\}, \{q^2, q^3\}\}$.
 Two other examples are $\operatorname{DNF}(\bot) = \operatorname{DNF}(\bigvee \emptyset) = \emptyset$ and $\operatorname{DNF}(\top) = \operatorname{DNF}(\bigvee \{\bigwedge \emptyset\}) = \{\emptyset\}$.
\end{example}

\begin{definition}
 \label{def:lts-induced}
 Let $s \in \alts$.
 Then the \emph{induced IA} of $s$ is defined as
 \begin{align*}
 \ltsop(s) \isdef& (\mathcal{P}(Q_s),I, O,T,\operatorname{DNF}(e_s^0)) \in \lts \text{, with for $Q \subseteq Q_s$ and $\ell \in L$:}\\
 T(Q,\ell) =&
  \begin{cases}
    \operatorname{DNF}((\bigwedge Q) [{T_s}_\ell]) \setminus \{\emptyset\} & \text{if $\ell \in I$}\\
    \operatorname{DNF}((\bigwedge Q) [{T_s}_\ell]) & \text{if $\ell \in O$}\\
  \end{cases}
 \end{align*}
\end{definition}

A state of $\ltsop(s)$ acts as the conjunction of the corresponding states in $s$.
In particular, a singleton state $\{q\}$ in $\ltsop(s)$ acts as the contained state $q$ in $s$, and state $\emptyset$ in $\ltsop(s)$ acts as a chaotic state, having $\irtraces_{\ltsop(s)}(\emptyset) = \iftracesdom_{I,O}$.

\begin{proposition}
  \label{pro:induced-ia-preserves-traces}
  Let $s \in \alts$.
  Then
  $\irtraces(s) = \rcl(\irtraces(\ltsop(s)))$.
\end{proposition}

\begin{proof}
  \todo{is this too short?}
  Analogously to Proposition~ \ref{pro:atraces-lts-alts}, using that
  \[\rcl(\irtraces_{\ltsop(s)}(\{q_1, \dots, q_n\})) = \irtraces_{s}(q_1 \wedge \dots \wedge q_n) \qedtag\]
\end{proof}

\begin{corollary}
  Let $s_1, s_2 \in \alts$.
  Then
  $s_1 \ir s_2 \iff \ltsop(s_1) \ir \ltsop(s_2)$
\end{corollary}

%% file: test-cases.tex
\section{Testing Input-Failure Refinement}
\label{sec:testers}

So far, we have introduced refinement as a way of specifying correctness of one model with respect to another.
Often, a specification is indeed a model, but we use it to ensure correctness of a real-world software implementation.
To this end, we assume that this implementation behaves like a IA.
We cannot see the actual states and transitions of this IA, but we can provide inputs to it and observe its outputs.
We assume that this IA must have an initial state, i.e. it is \emph{non-empty}.

\begin{definition}
  \label{def:ia-nonempty}
  \cite{AlHe01}
  An IA $i$ is \emph{empty} if $Q_i^0 = \emptyset$.
\end{definition}

In this section, we introduce a basis for \emph{model-based testing} with AIAs, analogously to ioco test case generation~\cite{Tre08}.
Given a specification AIA, we derive a testing experiment on non-empty implementation IAs, in order to observe whether input-failure refinement holds with respect to the specification.
This requires an extension of input-failure refinement to these domains.

\begin{definition}
  \label{def:ir-on-ia-aia}
  Let $i \in \lts$ and $s \in \alts$.
  Then 
  \[i \ir s \iff \irtraces(i) \subseteq \irtraces(s).\]
\end{definition}

\subsection{Testers for AIA Specifications}

From a given specification AIA, we derive a tester.
We model this tester as an IA as well, which can communicate with an implementation IA through a form of parallel composition.
The tester eventually concludes a verdict, indicating whether the observed behaviour is allowed.
To communicate, the inputs of the implementation must be outputs for the tester, and vice versa (note that $I$ and $O$ denote the inputs and outputs for the \emph{implementation}, respectively).
The tester should not block or ignore outputs from the implementation, meaning that the tester should be input-enabled.
If the tester intends to supply an input to the implementation, it should also be prepared for a refusal of that input.
A verdict is given by means of special states $\pass$ or $\fail$.
Lastly, to give consistent verdicts, a tester should be deterministic.
This leads to the following definition of testers.

\begin{definition}
  \label{def:tester}
  A \emph{tester for (an IA or AIA with) inputs $I$ and outputs $O$} is a deterministic, input-enabled IA $t = (Q_t,\; O,\; I \cup \overline{I},\; T,\; q_t^0)$ with $\pass,\fail \in Q_t$, such that $\pass$ and $\fail$ are sink-states with $\out(\pass) = \out(\fail) = \emptyset$, and $a \in \out(q) \iff \overline{a} \in \out(q)$ for all $q \in Q_t$ and $a \in I$.
\end{definition}

Testing is performed by a special form of parallel composition of a tester and an implementation.
If the tester chooses to perform an input while the implementation also chooses to produce an output, this results in a race condition.
In such a case, both the input or the output can occur during test execution.
We assume a synchronous setting, in which the implementation and specification agree on the order in which observed actions are performed (in contrast to e.g. a queue-based setting~\cite{PeYeHu03}, in which all possible orders are accounted for).
These assumptions are in line with the assumptions in e.g. ioco-theory~\cite{Tre08}, and lead to the following definition of test execution.

\begin{definition}
  \label{def:tester-execution}
  Let $i \in \lts$ be non-empty, and let $t$ be a tester for $i$.
  We write $q_t \te q_i$ for $(q_t, q_i) \in Q_t \times Q_i$.
  Then \emph{test execution of $i$ against $t$}, denoted $t \te i$, is defined as $(Q_t \times Q_i,\; \emptyset,\; I \cup \overline{I} \cup O,\; T,\; q_t^0 \te q_i^0) \in \lts$, with
  \begin{align*}
    T(q_t \te q_i, \ell) &= \{q_t' \te q_i' \:\mid\, q_t \xrightarrow{\ell} q_t',\; q_i \xrightarrow{\ell} q_i'\} &\text{for $\ell \in L$}\\
    T(q_t \te q_i,\; \overline{a}) &= \{q_t' \te q_i \:\mid\, q_t \xrightarrow{\overline{a}} q_t',\; q_i \not\xrightarrow{a}\} &\text{for $a \in I$}
  \end{align*}
  We say that $i \fails t$ if $q_t^0 \te q_i^0 \xrightarrow{\sigma} \fail \te q_i$ for some $\sigma$ and $q_i$, and $i \passes t$ otherwise.
\end{definition}

We reuse the notions of \emph{soundness} and \emph{exhaustiveness} from~\cite{Tre08}, to express whether a tester properly tests for a given specification.

\begin{definition}
  \label{def:soundness-exhaustiveness}
  Let $s \in \alts$ and let $t$ be a tester for $s$.
  Then $t$ is \emph{sound} for $s$ if for all $i \in \lts$ with inputs $I$ and outputs $O$, $i \fails t$ implies $i \not\ir s$.
  Moreover, $t$ is \emph{exhaustive} for $s$ if for all $i \in \lts$, $i \passes t$ implies $i \ir s$.
\end{definition}

A simple attempt to translate specification AIA $s$ to a sound and exhaustive tester would be similar to the determinization of $s$, but replacing every occurence of $\bot$ and $\top$ by $\fail$ and $\pass$, respectively.
\[f_t(e) =
  \begin{cases}
    \fail & \text{if $e = \bot$}\\
    \pass & \text{if $e = \top$}\\
    e & \text{otherwise}
  \end{cases}\]
Taking special care of input failures, the function $f_t$ then induces a tester $(\SE(Q_s) \cup \{\pass,\fail\}, O, I \cup \overline{I}, T, f_t(e_s^0))$, with
\begin{align*}
T(e,\ell) &\isdef \{f_t(\after[s]{e}{\ell})\} & \text{for $e \in \SE(Q_s), \ell \in L$}\\
T(v,\ell) &\isdef 
\begin{cases}
  \{v\} &\text{if $\ell \in O$}\\
  \emptyset  &\text{if $\ell \in I$}
\end{cases} 
&\text{for $v \in \{\pass,\fail\}$}\\
T(e,\overline{a}) &\isdef
\begin{cases}
  \{\pass\} & \text{if $(\after[s]{e}{a}) = \top$}\\
  \{\fail\} & \text{otherwise}
\end{cases}& \text{for $e \in \SE(Q_s), a \in I$}
\end{align*}
This tester is sound and complete for $s$: each possible input-failure trace  is in $\irtraces(s)$ if and only if it does not lead to $\fail$, by construction.
Here, we make use of the fact that $\irtraces(\bot) = \emptyset$, meaning that $\bot$ cannot be implemented correctly by a non-empty IA and can thus be replaced by $\fail$.
Likewise, $\irtraces(\top) = \iftracesdom_{I,O}$ means that $\top$ is always implemented correctly, and can be replaced by $\pass$.

However, this tester is quite inefficient.
If a tester reaches $\pass$ after both $\sigma a$ and $\sigma\overline{a}$, then this input $a$ does not need to be tested after $\sigma$.
Specifically, this is the case if and only if trace $\sigma a$ leads to specification configuration $\top$.
We thus improve the tester for a given specifications as follows.

\begin{definition}
  \label{def:tester-induced}
  Let $s \in \alts$.
  Then $\tester(s) \in \lts$ is defined as
  \[
    \tester(s) \isdef (\SE(Q_s) \cup \{\pass,\fail\}, O, I \cup \overline{I}, T, f_t(e_s^0)), \text{ with $f_t$ as before, and}
  \]
  \vspace{-6mm}
  \begin{align*}
    T(e,\ell) &\isdef 
    \begin{cases}
      \{f_t(\after[s]{e}{\ell})\} & \text{if $\ell \in O$, or $\ell \in I$ and $(\after[s]{e}{\ell}) \neq \top$}\\
      \emptyset & \text{if $\ell \in I$ and $(\after[s]{e}{\ell}) = \top$}
    \end{cases} &\text{\hspace{-1mm}for $\ell \in L$}\\
    T(e,\overline{a}) &\isdef
    \begin{cases}
      \emptyset & \text{if $(\after[s]{e}{a}) = \top$}\\
      \{\fail\} & \text{otherwise}
    \end{cases} &\text{\hspace{-38mm}for $e \in \SE(Q_s), a \in I$}\\
    T(v,\ell) &\isdef 
    \begin{cases}
      \{v\} &\text{if $\ell \in O$}\\
      \emptyset  &\text{if $\ell \in I$}
    \end{cases} & \text{\hspace{-38mm}for $v \in \{\pass,\fail\}, \ell \in L$}\\
%
  \end{align*}
\end{definition}

\begin{longversion}
  As expected, this tester tests precisely for input-failure refinement.
\end{longversion}

\begin{lemma}
  \label{lem:failing-is-not-refining}
  For $i \in \lts$ and $s \in \alts$, \;
  $i \fails \tester(s) \iff i \not\ir s$.
\end{lemma}

\begin{proof}
  \begin{align*}
         & i \fails \tester(s_1)\\
    \iff
         & \exists \sigma \in \iftracesdom_{I,O}: \tester(s_1) \te i \xrightarrow{\sigma} \fail \te q_i
         \ttag{Definition~\ref{def:tester-execution}}\\
    \iff 
         & \exists \sigma \in L^*: \tester(s_1) \xrightarrow{\sigma} \fail \text{ and } q_i^0 \xrightarrow{\sigma}\\
         & \text{or } \exists \sigma \overline{a} \in L^* \cdot \overline{I}: \tester(s_1) \xrightarrow{\sigma \overline{a}} \fail \text{ and } q_i^0 \xrightarrow{\sigma} q_i \not\xrightarrow{a}
         \ttag{Definitions~\ref{def:tester-execution} and~\ref{def:tester-induced}, formally by induction}\\
    \iff 
         & \exists \sigma \in L^*: (\after{s_1}{\sigma}) = \bot \text{ and } q_i^0 \xrightarrow{\sigma}\\
         & \text{or } \exists \sigma\overline{a} \in L^* \cdot \overline{I}: (\after{s_1}{\sigma a}) \neq \top \text{ and } q_i^0 \xrightarrow{\sigma} q_i \not\xrightarrow{a}
         \ttag{Construction of $\tester$ in Definition~\ref{def:tester-induced}, formally by induction}\\
    \iff 
         & \exists \sigma \in L^*: \sigma \not\in \irtraces(s_1) \text{ and } q_i^0 \xrightarrow{\sigma}\\
         & \text{or } \exists \sigma\overline{a} \in L^* \cdot \overline{I}: \sigma \overline{a} \not\in \irtraces(s_1) \text{ and } q_i^0 \xrightarrow{\sigma} q_i \not\xrightarrow{a}
         \ttag{Definition~\ref{def:atraces-alts}}\\
    \iff 
         & \exists \sigma \in \iftracesdom_{I,O}: \sigma \not\in \irtraces(s_1) \text{ and } \sigma \in \irtraces(i)
         \ttag{Definition~\ref{def:ir-inclusion}}\\
    \iff 
         & \irtraces(i) \not\subseteq \irtraces(s_1)
         \ttag{Set theory}\\
    \iff
         & i \not\ir s_1
         \ttag{Definition~\ref{def:ir-on-ia-aia}}
  \end{align*}
  \strut\qed
\end{proof}

\begin{example}
  The tester for $s_B$ in Figure~\ref{fig:ALTS-coffee} is shown in Figure~\ref{fig:tester}.
  
  \begin{figure}
    \begin{center}
      \begin{tikzpicture}[->,node distance=18mm]
        \tikzlts
        \node[initial left,state] (init) {$q_B^0$};
        \node (init-text) [above left=-2mm and 0mm of init] {$\tester(s_B)$};
        \node[recstate] (on) [below=8mm of init] {$q_B^{1} \wedge q_B^3 \wedge q_B^5 \wedge q_B^8$};
        \node[recstate] (A) [below left=8mm and -4mm of on] {$q_B^{2} \wedge q_B^9$};
        \node[recstate] (B) [below right=8mm and -12mm of on] {$(q_B^4 \wedge (q_B^6 \vee q_B^7) \wedge q_B^9$};
        \node[state] (take) [below=18mm of on] {$q_B^{10}$};
        \node[verdictstate] (fail) [below right=15mm and 5mm of take] {$\fail$};
        \path
        (init) edge node [left] {!on} (on)
        (on) edge node [left=1mm] {!A} (A)
        (on) edge node [right=1mm] {!B} (B)
        (A) edge node [below left=-1mm] {?c} (take)
        (B) edge node [below right=-1mm and -0.5mm] {?t+m} (take)
        (take) edge node [left=-0.8mm] {!take} (on)
        
        (init) edge [in=0,out=-5, looseness=2.9] node [align=left, right] {?O\\ !$\overline{\text{on}}$} (fail)
        (on) edge [in=15,out=-10, looseness=2.6] node [align=left, right] {?O\smallskip\\!$\overline{\text{A}}$\smallskip\\!$\overline{\text{B}}$} (fail)
        (B) edge [in=30,out=-60, looseness=1.2] node [align=left, above right=0mm and 2.5mm] {?t\\?c\\?c+m} (fail)
        (take) edge [bend left=20] node [align=left, above right=-4mm and 0.5mm] {?O\bigskip\\!$\overline{\text{take}}$} (fail)
        (A) edge [in=160,out=-80, looseness=1.0] node [align=right, below left=-10mm and 5mm] {?t\\?t+m\\?c+m} (fail)
        (fail) edge [loop below] node {?O} ()
        ;
      \end{tikzpicture}
      \caption{The tester for the vending machine.
      The label ?O denotes a transition for every label in O.
      Remark that inputs for $s_B$ are outputs for $\tester(s_B)$, and vice versa.}
      \label{fig:tester}
    \end{center}
  \end{figure}
\end{example}

Theorem~\ref{the:soundness-exhaustiveness} shows that soundness and exhaustiveness of a tester corresponds to refinement of the corresponding AIA.

\begin{theorem}
  \label{the:soundness-exhaustiveness}
  Let $s_1, s_2 \in \alts$.
  Then
  
  \def\arraystretch{1.3}
  \begin{tabular}{lrl}
    1 & \multicolumn{2}{l}{\;$\tester(s_1)$ is sound and exhaustive for $\ltsop(s_1)$}\\
    2 &\; $\tester(s_1)$ is sound for $s_2$ &       $\iff s_2 \ir s_1$ \label{the:soundness-exhaustiveness:soundness}\\
    3 &\; $\tester(s_1)$ is exhaustive for $s_2$ &  $\iff s_1 \ir s_2$
    \label{the:soundness-exhaustiveness:exhaustiveness}
  \end{tabular}
\end{theorem}

\begin{proof}
  Lemma~\ref{lem:failing-is-not-refining} implies that $\tester(s_1)$ is sound and exhaustive for any AIA $s$ with $\irtraces(s) = \irtraces(s_1)$.
  This proves that $\tester(s_1)$ is indeed sound and exhaustive for $s_1$ itself, as well.
  
  Now, we prove the second statement.
  This statement is vacuous if $e_2^0 = \bot$, so assume $e_2^0 \neq \bot$, which implies that $\ltsop(s_2)$ is defined.
  Then
  \begin{align*}
         & \text{$\tester(s_1)$ is sound for $s_2$}\\
    \iff
         & \text{for all IAs $i$}, i \fails \tester(s_1) \implies i \not\ir s_2\\
    \iff
         & \text{for all IAs $i$}, i \not\ir s_1 \implies i \not\ir s_2\\
    \iff
         & \text{for all IAs $i$ with $i \ir s_2$}, \irtraces(i) \subseteq \irtraces(s_1)\\
    \iff
         & \bigcup \{\irtraces(i) \mid i \in \lts, i \ir s_2\} \subseteq \irtraces(s_1)\\
    \iff
         & \irtraces(\ltsop(s_2)) \subseteq \irtraces(s_1)\\
    \iff
         & \irtraces(s_2) \subseteq \irtraces(s_1)\\
    \iff
         & s_2 \ir s_1
  \end{align*}
  The proof of the third statement is analogous.
  \qed
\end{proof}


\subsection{Test Cases for AIA Specifications}

In~\cite{Tre08}, an algorithm was introduced to generate \emph{test cases}.
These are testers as in Definition~\ref{def:tester} with additional restrictions, so that they can be used as unambiguous instructions to test a system.
In particular, states of a test case should have at most one outgoing input transition.
This ensures that no choice between different inputs has to be resolved during test execution.
Additionaly, all paths of a test case lead to $\pass$ or $\fail$ in a finite number of steps, to ensure that test execution terminates with a verdict.

\begin{definition}
  A tester $t$ for $I$ and $O$ is a \emph{test case} if
  \begin{itemize}
    \item
      for all $q_t \in Q_t$, $|\out(q_t)| \le 1$, and
    \item
      there are no infinite sequences $q_t^0, q_t^1,\dots$ for ${q_t^0, q_t^1,\dots} \in Q_t \setminus \{\pass,\fail\}$ such that $q_t^0 \xrightarrow{\ell^0} q_t^1 \xrightarrow{\ell^1} \dots$
  \end{itemize}
\end{definition}

The test case generation algorithm of~\cite{Tre08} is non-deterministic, since it must choose at most one inputs in every state, and it must choose when to stop testing.
We avoid defining a separate test case generation algorithm, and instead use Theorem~\ref{the:soundness-exhaustiveness} to obtain sound test cases.
If specification $s_1$ is weakened to $s_2$, such that $\tester(s_2)$ is a test case, then soundness of $\tester(s_2)$ for $s_1$ is guaranteed by the theorem.
Such a weakened \emph{singular specification} $s_2$ describes a finite, tree-shaped part of the original specification $s_1$.

\begin{definition}
  \label{def:singular-specification}
  Let $s_1, s_2 \in \alts$.
  Then $s_2$ is a \emph{singular specification for $s_1$} if $Q_2$ is a finite subset of $L^*$, with $e_2^0 \in \{\epsilon, \top, \bot\}$, $e_1^0 = \top \implies e_2^0 = \top$ and $e_2^0 = \bot \implies e_1^0 = \bot$,
  and having that for every $\sigma \in Q_2$, the following holds:
  \begin{enumerate}
    \item
      $T_2(\sigma,\ell) = \bot \implies (\after{s_1}{\sigma\ell}) = \bot$ for $\ell \in L$,
    \item
      $(\after{s_1}{\sigma\ell}) = \top \implies T_2(\sigma,\ell) = \top$ for $\ell \in L$
    \item
      $T_2(\sigma,\ell)$ is either $\bot$ or $\top$ or $\sigma\ell$ for $\ell \in L$, and
    \item
      there is at most one $a \in I$ with $T(\sigma, a) \neq \top$.
  \end{enumerate}
\end{definition}

It can be created from $s_1$ similarly to test case generation in~\cite{Tre08}.
In every state $\sigma$ of the tree $s_1$, we either decide to pick one input specified in $s_1$ and also specify that in $s_2$; or we do not specify any input, but only outputs; or we leave any successive behaviour unspecified ($\top$).

Test cases based on singular specifications are inherently sound, and for any incorrect implementation, it is possible to find a singular specification which induces a test case to detects this incorrectness.

\begin{lemma}
  \label{lem:singular-is-weaker}
  If $s_2$ is a singular specification for $s_1$, then $s_1 \ir s_2$.
\end{lemma}

\begin{proof}
  By Definition~\ref{def:atraces-alts}, we prove that any $\sigma \in \irtraces(s_1)$ also has $\sigma \in \irtraces(s_2)$, by induction on the length of $\sigma$.
  
  The base case $\sigma = \epsilon$ trivially holds if $e_2^0 = \top$ or $e_1^0 = \bot$, so assume that $e_2^0 \neq \top$ and $e_1^0 \neq \bot$.
  In that case, Definition~\ref{def:singular-specification} implies $e_2^0 = \epsilon$.
  Then $\epsilon \in \irtraces(s_1)$ and $\epsilon \in \irtraces(s_2)$, proving the base case.
  
  For the inductive step, let $\sigma = \sigma'\ell$ with $\ell \in L \cdot \overline{I}$, and assume as induction hypothesis that $\sigma' \in \irtraces(s_1) \implies \sigma' \in \irtraces(s_2)$ (IH).
  From $\sigma \in \irtraces(s_1)$, we have that $\sigma' \in \irtraces(s_1)$ and therefore also $\sigma' \in \irtraces(s_2)$ (1) by (IH).
  Now, we make a case distinction:
  \begin{itemize}
    \item
      If $\ell \in L$, then $\sigma \in \irtraces(s_1)$ implies $\after{s_1}{\sigma} \neq \bot$ by Definition~\ref{def:atraces-alts}, and therefore $T_2(\sigma',\ell) \neq \bot$ by Definition~\ref{def:singular-specification}.
      This implies $\sigma \in \irtraces(s_2)$ by Definition~\ref{def:atraces-alts}.
    \item
      If $\ell = \overline{a} \in \overline{I}$, then $\sigma \in \irtraces(s_1)$ implies $\after{s_1}{\sigma} = \top$ by Definition~\ref{def:atraces-alts}, and therefore $T_2(\sigma',\ell) = \top$ by Definition~\ref{def:singular-specification}.
      This implies $\sigma \in \irtraces(s_2)$ by Definition~\ref{def:atraces-alts}.
      \qed
  \end{itemize}
\end{proof}

\begin{theorem}
  If $s_2$ is a singular specification for $s_1$, then $\tester(s_2)$ is a sound test case for $s_1$.
\end{theorem}

\begin{proof}
  This follows directly from Lemma~\ref{lem:singular-is-weaker} and Theorem \ref{the:soundness-exhaustiveness}.
  \qed
\end{proof}

\begin{theorem}
  Let $i \in \lts$ and $s_1 \in \alts$.
  If $i \not\ir s_1$, then there is a singular specification $s_2$ for $s_1$ such that $i \fails \tester(s_2)$.
\end{theorem}

\begin{proof}
  If $i \not\ir s_1$, then there is some $\sigma \in \iftracesdom_{I,O}$ with $\sigma \in \irtraces(i)\setminus\irtraces(s_1)$.
  We construct $s_2$ as follows:
  \begin{itemize}
    \item
      If $\sigma = \epsilon$, then choose $s_2 = (\emptyset,I,O,T_\emptyset,\bot)$ where $T_\emptyset$ is the empty function.
    \item
      If $\sigma = \ell_0\dots\ell_n$ for some $n \ge 0$, then choose $s_2 = (\{\ell_0\dots\ell_i \mid 0 \le i \le n\},I,O,T,\epsilon)$ with $T(\ell_0\dots\ell_i,\ell_{i+1}) = \ell_0\dots\ell_{i+1}$ for $0 \le i < n$, and 
      \[T(\ell_0\dots\ell_{n-1}, \ell_n) =
        \begin{cases}
          \bot & \text{if $\ell_n \in O$}\\
          \ell_0\dots a & \text{if $\ell_n = \overline{a} \in \overline{I}$}
        \end{cases}\]
      and $T(q,\ell) = \top$ for any other pair $(q,\ell) \in Q_s \times L$.
  \end{itemize}
  By construction, both these definitions of $s_2$ are singular for $s_1$ and have $\sigma \not \in \irtraces(s_2)$.
  Since $\sigma \in \irtraces(i)$, Definition~\ref{def:ir-on-ia-aia} implies that $i \not\ir s_2$, and therefore $\tester(s_2) \fails s_1$ holds by Theorem~\ref{the:soundness-exhaustiveness}.\ref{the:soundness-exhaustiveness:exhaustiveness}.
  \qed
\end{proof}


\begin{example}
  Specification $s_B$ in Figure~\ref{fig:alts-example} can be weakened to singular specification $s_C$ shown in Figure~\ref{fig:test-case}.
  Indeed, $s_B \ir s_C$ holds, which can be established by comparing $s_C$ with $\det(s_B)$ in Figure~\ref{fig:example-determinizations}.
  Therefore $\tester(s_C)$ is a sound test case for $s_B$.
  
  \begin{figure}
    \begin{center}
      \begin{tikzpicture}[->,node distance=10.42mm]
        \tikzlts
        \node[initial left,emptystate] (0) {};
        \node (init-text) [above left=-1mm and 0mm of 0] {$s_C$};
        \node[emptystate] (1) [below of=0] {};
        \node[emptystate] (2) [below of=1] {};
        \node[emptystate] (3) [below of=2] {};
        \node[emptystate] (4) [below of=3] {};
        \node[emptystate] (5) [below of=4] {};
        \node[top] (6) [below of=5] {};
        \path
        (0) edge node [left] {?on} (1)
        (1) edge node [left] {?a} (2)
        (2) edge node [left] {!c} (3)
        (3) edge node [left] {?take} (4)
        (4) edge node [left] {?b} (5)
        (5) edge node [left] {!t+m} (6)
        ;
      \end{tikzpicture}
      \hspace{10mm}
      \begin{tikzpicture}[->,node distance=10.42mm]
        \tikzlts
        \node[initial left,emptystate] (0) {};
        \node (init-text) [above left=-1mm and 0mm of 0] {$\tester(s_C)$};
        \node[emptystate] (1) [below of=0] {};
        \node[emptystate] (2) [below of=1] {};
        \node[emptystate] (3) [below of=2] {};
        \node[emptystate] (4) [below of=3] {};
        \node[emptystate] (5) [below of=4] {};
        \node[verdictstate] (pass) [below of=5] {$\pass$};
        \node[verdictstate] (fail) [right=30mm of 3] {$\fail$};
        \path
        (0) edge node [left] {!on} (1)
        (1) edge node [left] {!a} (2)
        (2) edge node [left] {?c} (3)
        (3) edge node [left] {!take} (4)
        (4) edge node [left] {!b} (5)
        (5) edge node [left] {?t+m} (pass)
        
        (0) edge [in=90,out=0] node [above right=-1mm, align=left] {!$\overline{\text{on}}$\\?O} (fail)
        (1) edge [in=120,out=0] node [pos=0.4,above=1mm, align=left] {!$\overline{\text{a}}$ ?O} (fail)
        (2) edge [in=150,out=0] node [pos=0.3,above=0mm, align=left] {?t ?t+m\\?c+m} (fail)
        (3) edge [in=180,out=0] node [pos=0.3,above=0mm, align=left] {!$\overline{\text{take}}$ ?O} (fail)
        (4) edge [in=210,out=0] node [pos=0.3,above=1mm, align=left] {!$\overline{\text{b}}$ ?O} (fail)
        (5) edge [in=240,out=0] node [pos=0.3,above=1mm, align=left] {?t ?c\\?c+m} (fail)
        (pass) edge [loop right] node {?O} (pass)
        (fail) edge [loop right] node {?O} (fail)
        ;
      \end{tikzpicture}
      \caption{A weakened version $s_C$ of the vending machine, and the test case $\tester(s_C)$.
      Question and exclamation marks are interchanged in $\tester(s_C)$ to indicate that the input and output alphabets have been interchanged with respect to $s_C$.
      }
      \label{fig:test-case}
    \end{center}
  \end{figure}
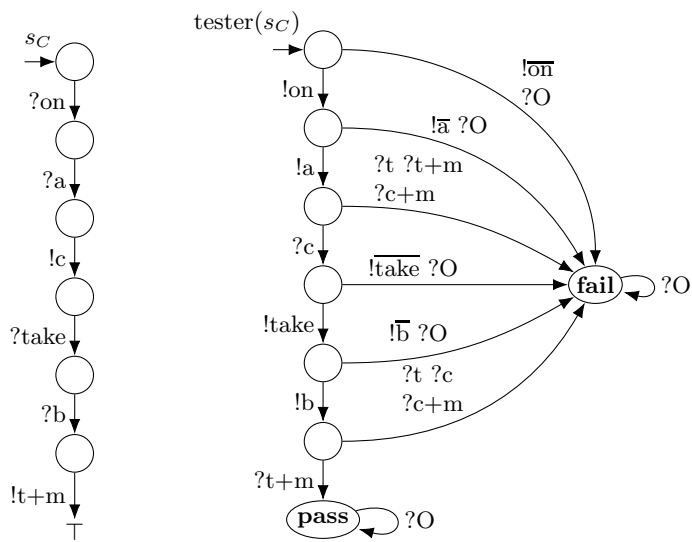
\end{example}

%
%
%
%
%

%% file: conclusion.tex
\section{Conclusion and Future Work}

Alternating interface automata serve as a natural and direct representation for sets of input-failure traces, and therefore also for refinement of systems with inputs, outputs, non-determinism and conjunction.
We have used the observational nature of input-failure traces to define testers, describing an experiment to observationally establish refinement of a black-box system.

The disjunction and conjunction of alternation brings interface automata specifications closer to the realm of logic and lattice theory.
On the theoretical side, a possible direction is to extend configurations from distributive lattices to a full logic.
On the practical side, classical testing techniques acting on logical expressions, such as combinatorial testing, could be translated to our black-box configurations of states.

Possible criticism on our running example of a vending machine $s_B$ in Figure~\ref{fig:ALTS-coffee} may be that its representation as an AIA is not concise, since the determinization $\det(s_B)$ is much smaller and more understandable than $s_B$ itself.
This is because the individual specifications offer a choice between outputs, such as tea with or without milk, whereas the intersection of all choices is singleton.
A more natural encoding for this example is to express the types of drink with data data parameters, and the restrictions on them by logical constraints.
This requires an automaton model in style of symbolic transition systems~\cite{FrTr07}, which could be enriched with the concepts of alternation of AIAs.

Interface automata typically contain internal transitions, and the interaction between internal behaviour and alternation is not immediately clear.
A possible approach to extend AIAs with internal behaviour is to lift the $\epsilon$-closure of~\cite{AlHe01}, the set of states reachable via internal transitions, to the level of configurations.

%

